\documentclass[a4paper, 11pt]{article}
\usepackage[top=1.0 in,bottom=1.0 in,left=1.0 in,right=1.0 in]{geometry}
\usepackage{amsmath,amsthm}
\usepackage{amssymb}
\usepackage[colorlinks,linkcolor=blue,citecolor=blue]{hyperref}
\usepackage{float}
\usepackage{cite}
\usepackage{epic}
\usepackage{tikz}
\usepackage{graphicx}
\usepackage{accents}
\usepackage{slashed}

\newtheorem{lemma}{Lemma}[section]
\newtheorem{theorem}{Theorem}[section]

\numberwithin{equation}{section}
\newtheorem{remark}{Remark}[section]

\makeatletter

\newcommand{\Rmnum}[1]{\expandafter\@slowromancap\romannumeral #1@} \makeatother

\title{The Lam\'e functions and elliptic soliton solutions: Bilinear approach}

\author{Xing Li$^{1}$, ~~ Da-jun Zhang$^{2}$\\
${}^1$
  School of Mathematics and Statistics, Jiangsu Normal University,\\
   Xuzhou, Jiangsu 221116, P.R. China\\
 ${}^2$
  Department of Mathematics,
		Shanghai University,\\
 Shanghai 200444,  P.R. China\\
 E-mail: xing$\underline{~}$li@jsnu.edu.cn \quad
djzhang@staff.shu.edu.cn\\
}
\date{}

\begin{document}

\maketitle

\begin{abstract}
 The Lam\'e function can be used to construct plane wave factors and  solutions to the Korteweg-de Vries (KdV) and Kadomtsev-Petviashvili (KP) hierarchy. The solutions are usually called elliptic solitons.
 In this chapter, first, we review recent development in the Hirota bilinear method on elliptic solitons
 of the KdV equation and KP equation, including bilinear calculations involved with  the Lam\'e type plane wave factors,
 expressions of $\tau$ functions and the generating vertex operators.
 Then, for the discrete potential KdV and KP equations,
 we give their bilinear forms, derive $\tau$ functions of elliptic solitons,
 and show that they share the same vertex operators with the
 KdV hierarchy and the KP hierarchy, respectively.

\vskip 6pt

\noindent
\textbf{Key Words:}  Lam\'e function, Weierstrass function,  elliptic soliton solution,  $\tau$ function, vertex operator.
\end{abstract}

\tableofcontents

\vskip 20pt
\section{Introduction}

The Sato theory of integrable systems, established by Sato and his collaborators from Kyoto group,
provides deep insight into integrable systems (see \cite{MJD-book-1999} and the references therein).
The Kadomtsev-Petviashvili (KP) hierarchy  serves as a fundamental role in Sato's theory,
of which the  solutions form an infinite-dimensional Grassmann manifold.
The Pl\"ucker relation of the manifold corresponds to the bilinear  equations
satisfied  by the so-called $\tau$ function, which can also be defined using vertex operators.

Elliptic curves may play roles in integrable systems, either as elliptic type solutions
or as elliptic deformations of the equations themselves.
Both ways bring richer insight to integrable systems.
In order to solve the problem of periodic boundary condition
of the Korteweg-de Vries (KdV) equation
 \begin{equation}\label{eq:kdv}
u_t=\frac{3}{2}uu_x+\frac{1}{4}u_{xxx},
\end{equation}
Novikov, Matveev, et al
established the  finite-gap integration theory (see \cite{Matveev-08} and the references therein).
Considering the Lax pair of the KdV equation  \eqref{eq:kdv}:
 \begin{subequations}\label{eq:lax-kdv}
\begin{align}\label{eq:lame}
 \varphi_{xx}&=(\lambda-u)\varphi ,\\
 \varphi_{t}&=\mathbf{\varphi}_{xxx}+\frac{3}{2}u\varphi_{x}+\frac{3}{4}u_x\varphi,
\end{align}
\end{subequations}
where $u(x,0)$ is assumed to be a periodic function of $x$,
the spectrum $\lambda$ generally corresponds to the values of
a finite sequence of segments on the real line,
and the recovered $u(x,t)$ is  called finite-gap solution.
In particular, $u=-2\wp(x)$ is called a one-gap solution, where $\wp$ is the Weierstrass $\wp$  function,
which is elliptic.
With an elliptic potential $u=-2\wp(x)$ and $\lambda=\wp(k)$,
\eqref{eq:lame} is also known as  a Lam\'e equation \cite{Ince-1940},
whose two linear independent solutions,  known as the Lam\'e functions, are $\Psi_x(k)$ and $\Psi_x(-k)$,
where
\begin{equation}\label{Phi}
  \Psi_x(k)=\Phi_x(k)e^{-\zeta(k)x},\qquad \Phi_x(k)=\frac{\sigma(x+k)}{\sigma(x)\sigma(k)}.
\end{equation}
Here $\sigma$ and $\zeta$ are the Weierstrass  functions $\sigma$ and $\zeta$.
Inserting time evolution, the Lam\'e function  can be written as
 \begin{subequations}\label{PFW}
 \begin{equation}
  \varphi(x,t;k)=a^+\varphi^+(x,t;k) +a^-\varphi^-(x,t;k),
 \end{equation}
where
\begin{equation}\label{phi-j-gamma}
\varphi^{\pm}(x,t;k)=\Phi_{x}(\pm k) e^{\mp \gamma(k)},~~ \gamma(k)=\zeta(k)x-\frac{1}{2}\wp'(k)t+\gamma^{(0)}(k), ~~\gamma^{(0)}(k)\in\mathbb{C}.
\end{equation}
  \end{subequations}
In this chapter we call \eqref{PFW} a Lam\'e type plane wave factor (PWF), and it can be regarded as the wave function of the KdV equation \eqref{eq:kdv}.
In 1974, for the KdV equation, Kuznetsov and Mikhailov, using the Marchenko integral equation,
first derived the  multisoliton solutions composed by the above PWF \cite{74KM}.
In the discrete case, Nijhoff and Atkinson \cite{IMRN2010} developed a Cauchy matrix approach to
the Adler-Bobenko-Suris \cite{ABS-2003} lattice equations (except Q4 equation)
which are the discrete KdV type.
In their approach they use discrete (analogue of) Lam\'e type PFWs
and the obtained solutions are  termed as  ``elliptic $N$-soliton solutions''.
Later, their approach was applied to the lattice (potential) KP equation \cite{YN-JMP-2013}.
We will follow this name of \textit{elliptic solitons} in our paper.
More recently, a remarkable progress of such type solutions
is the establishment of an elliptic scheme of the direct linearisation approach \cite{22NSZ},
in which the concept of elliptic $N$-th roots of unity was introduced to define PFWs
of the discrete Boussinesq type equations as well as implementing dimension reductions.
Another progress is about the bilinear framework related to the Lam\'e type PFWs \cite{22LxZhang}.
We explored some new properties of bilinear derivation calculation when  the usual exponential PWFs
are replaced with the Lam\'e type ones. Based on these properties,
$\tau$ functions, vertex operators and bilinear identities of the KdV and KP hierarchies were constructed,
 dimension reduction and period degeneration were also investigated \cite{22LxZhang}.

In this chapter, we will extend our investigation to the discrete case,
and derive the $\tau$ functions and vertex operators for the elliptic solitons of
the discrete potential KdV equation and KP equation.
This chapter is organized as follows.
After providing preliminaries in Sec.\ref{sec-2},
we will, in Sec.\ref{sec-3},  partially review the results about the $\tau$ functions and vertex operators for the elliptic solitons of the KdV equation and KP equation.
Then in Sec.\ref{lkdv}, we will verify the elliptic soliton solution in Casoratian form for the lattice potential KdV equation,  and further we will discuss  the $\tau$ function and vertex operator through the so-called
Miwa coordinates correspondence. The relevant results of the lattice potential KP equation will be sketched in Sec.\ref{lkp}.

\section{Preliminary}\label{sec-2}

Let us recall the Weierstrass functions, which are $\zeta(z)$, $\wp(z)$ and $\sigma(z)$.
Among them only $\wp(z)$ is  elliptic,
i.e. meromorphic and doubly periodic.
$\wp(z)$ is even, while the other two are odd and quasi-doubly periodic.
The periodic behaviors are the following:
\begin{subequations}\label{periodicity}
\begin{align}
&\wp(z+2w_i)=\wp(z), \quad i=1,2, \\
&\zeta(z+2 w_i)=\zeta(z)+2 \zeta(w_i), \\
&\sigma(z+2 w_i)=-\sigma(z)e^{2\zeta(w_i)(z+w_i)}.
\end{align}
\end{subequations}
$w_1$ and $w_2$ are half periods of $\wp(z)$ and satisify $\mathrm{Im}\frac{w_1}{w_2}\neq 0$.
The two periods $2w_1$ and $2w_2$ form the periodic lattice $\mathbb{D}$, and $\wp(z)$ is well-defined on the torus $\mathbb{C/D}$.
Let $e_i=\wp(w_i)$ for $i=1,2,3$, where $w_3=w_1+w_2$.
$(\wp(z), \wp'(z))$ is a point on the Weierstrass elliptic curve
\begin{subequations}
\begin{equation}\label{ell-cur}
y^2=4x^3-g_2x-g_3=4(x-e_1)(x-e_2)(x-e_3),
\end{equation}
i.e.
\begin{equation}\label{ell-cur-2}
(\wp'(z))^2=4\wp^3(z)-g_2\wp(z)-g_3,
\end{equation}
\end{subequations}
where
$g_2=-4(e_1e_2+e_2e_3+e_3e_1)$ and $g_3=4e_1e_2e_3$ are invariants of the curve.
This relation is often  used when calculating higher order derivatives of $\wp(z)$.
Three Weierstrass functions are connected via
\[\zeta(z)=\frac{\sigma'(z)}{\sigma(z)}, \qquad  \wp(z)=-\zeta'(z).\]
In addition, the basic addition rules are given below:
\begin{equation}\label{eq:add-1}
 \wp(z)-\wp(u)=-\frac{\sigma(z+u)\sigma(z-u)}{\sigma^2(z)\sigma^2(u)}=\Phi_{z}(u)\Phi_z(-u),
\end{equation}
where $\Phi_z(u)$ is defined in \eqref{Phi},
\begin{equation}\label{eq:add-2}
\eta_u(z)=\zeta(z+u)-\zeta(z)-\zeta(u)=\frac{1}{2}\frac{\wp'(z)-\wp'(u)}{\wp(z)-\wp(u)},
\end{equation}
\begin{equation}\label{eq:add-3}
\wp(z)+\wp(u)+\wp(z+u)=\eta_u^2(z),
\end{equation}
\begin{equation}\label{eq:add-4}
\chi_{u,v}(z)=\zeta(u)+\zeta(v)+\zeta(z)-\zeta(u+v+z)
=\frac{\sigma(u+v)\sigma(u+z)\sigma(z+v)}{\sigma(u)\sigma(v)\sigma(z)\sigma(z+u+v)},
\end{equation}
and
\begin{equation}\label{eq:add-5}
		\begin{split}
    & \sigma(x+a)\sigma(x-a)\sigma(y+b)\sigma(y-b)- \sigma(x+b)\sigma(x-b)\sigma(y+a)\sigma(y-a)\\
	 &= \sigma(x+y)\sigma(x-y)\sigma(a+b)\sigma(a-b).
		\end{split}
\end{equation}
The famous Frobenius-Stickelberger determinant (also known as elliptic van der Monde determinant) is \cite{FS80}
\begin{equation}\label{eq:FS-1}
	\begin{split}
	&|\mathbf{1},~ \wp(\mathbf{k}),~ \wp'(\mathbf{k}),~ \wp''(\mathbf{k}),~ \cdots,~\wp^{(n-2)}(\mathbf{k}) |\\
	=&(-1)^{\frac{(n-1)(n-2)}{2}}\left(\prod^{n-1}_{s=1}s!\right) \frac{\sigma(k_1+\cdots+k_n)
        \prod_{i<j}\sigma(k_i-k_j)}{\sigma^n(k_1)\sigma^n(k_2)\cdots\sigma^n(k_n)},
    \end{split}
\end{equation}
where $\mathbf{1}$ denotes a n-th order  column vector $(1, 1, \cdots, 1)^T$,
and $f{(\mathbf{k})}$ denotes a column vector with entries $f(k_j)$,
i.e.  $f{(\mathbf{k})}=(f(k_1), f(k_2), \cdots, f(k_n))^T$.
With a limit process, it follows from   \eqref{eq:FS-1} that (cf.\cite{Kiepert-1873})
\begin{equation}\label{eq:FS-2}
\begin{split}
(-1)^{n-1}\left(1!2!\cdots(n-1)!\right)^2\frac{\sigma(n\nu)}{\sigma(\nu)^{n^2}}=\left|\begin{array}{cccc}
 \wp'(\nu) & \cdots  & \wp^{(n-1)}(\nu)\\
\wp''(\nu) & \cdots  & \wp^{(n)}(\nu)\\
\vdots   & \vdots  & \vdots \\
 \wp^{(n-1)}(\nu) & \cdots  & \wp^{(2n-3)}(\nu)
\end{array}\right|.
\end{split}
\end{equation}
Other  useful formulas are (see (C.7) in \cite{ND-2016} and \cite{FS80})
\begin{equation}\label{eq:ND16}
\prod_{j=1}^n\Phi_x(k_j)=\frac{(-1)^{n-1}}{(n-1)!}\Phi_x(k_1+\cdots+k_n)
\frac{|\mathbf{1},~ \wp(\mathbf{k}),~ \wp'(\mathbf{k}),~ \cdots,~\wp^{(n-2)}(\mathbf{k}) |}
{|\mathbf{1}, ~\eta_x(\mathbf{k}),~ \wp(\mathbf{k}),~ \wp'(\mathbf{k}),~ \cdots,~\wp^{(n-3)}(\mathbf{k}) |},
\end{equation}
and
\begin{equation}\label{eq:add-6}
\begin{split}
n!\frac{\sigma(\nu-\mu)^{n}\sigma(\mu+n\nu)}{\sigma(\mu)^{n+1}\sigma(\nu)^n\sigma(n\nu)}&=\frac{\left|\begin{array}{cccc}
\wp(\mu)-\wp(\nu) & \wp'(\nu) & \cdots  & \wp^{(n-1)}(\nu)\\
\wp'(\mu)-\wp'(\nu) & \wp''(\nu) & \cdots  & \wp^{(n)}(\nu)\\
\vdots &   & \vdots  & \vdots \\
\wp^{(n-1)}(\mu)-\wp^{(n-1)}(\nu) & \wp^{(n)}(\nu) & \cdots  & \wp^{(2n-2)}(\nu)
\end{array}\right|}
{\left|\begin{array}{cccc}
 \wp'(\nu) & \cdots  & \wp^{(n-1)}(\nu)\\
\wp''(\nu) & \cdots  & \wp^{(n)}(\nu)\\
\vdots   & \vdots  & \vdots \\
 \wp^{(n-1)}(\nu) & \cdots  & \wp^{(2n-3)}(\nu)
\end{array}\right|}.
\end{split}
\end{equation}

Next, we introduce the Hirota operator $D$ \cite{74Hirota}
\begin{equation}\label{eq:hirota-D}
\begin{split}
e^{aD_x+bD_t}f(x,t)\cdot g(x,t)&=e^{a(\partial_x-\partial_{x'})+b(\partial_t-\partial_{t'})}f(x,t)\cdot g(x',t')|_{x'=x,t'=t},\\
&=f(x+a,t+b)g(x-a,t-b),
\end{split}
\end{equation}
where $a,b$ are arbitrary numbers, and $f$, $g$ are sufficiently smooth function of $(x,t)$.
Implementing Taylor expansion at $(a,b)=(0,0)$  and comparing  the coefficients of $a^nb^m$, we can obtain the equivalent definition of Hirota's operator (which is also known as Hirota's bilinear derivative operator)
\begin{equation}
D_x^nD_t^m f(x,t)\cdot g(x,t)=(\partial_x-\partial_{x'})^n(\partial_t-\partial_{t'})^mf(x,t)g(x',t')|_{x'=x,t'=t}.
\end{equation}
For $F$  being a polynomial of $D_x, D_t$, a bilinear derivative equation is defined by
\begin{equation}\label{eq:HB-cont}
F(D_x,D_t)f(x,t)\cdot g(x,t)=0.
\end{equation}
A discrete  analogue of the above equation
is defined by \cite{09HZ}
\begin{equation}\label{eq:HB-disc}
\sum_{j}c_jf_j(n+\nu_j^+,m+\mu_j^+)g_j(n+\nu_j^-,m+\mu_j^-)=0,
\end{equation}
 where $\nu_j^+ +\nu_j^-=\nu^s$,  $\mu_j^+ +\mu_j^-=\mu^s$, and $\nu^s$ and $\mu^s$ are independent of $j$.

 The fundamental property of the Hirota bilinear form is the  gauge property.
\begin{lemma}\label{L-gauge}
Let
\begin{equation}\label{f'g'}
f'(x,t)=e^{ax+bt}f(x,t),\quad g'(x,t)=e^{ax+bt}g(x,t).
\end{equation}
Then it holds that
 \begin{equation*}
D_x^nD_t^m f'(x,t)\cdot g'(x,t)=e^{2(ax+bt)}D_x^nD_t^mf(x,t)\cdot g(x,t).
 \end{equation*}
\end{lemma}

This means the bilinear equation \eqref{eq:HB-cont} is invariant under the transformation \eqref{f'g'}, i.e.
$F(D_x,D_t)f'(x,t)\cdot g'(x,t)=0$.
In discrete case, equation \eqref{eq:HB-disc} is invariant with the transformation
 $$f_j'(n,m)=A^nB^mf_j(n,m), \quad g_j'(n,m)=A^nB^mg_j(n,m),$$
where $A$ and $B$ are constants.
The gauge property will be changed when the usual PWFs (linear exponential functions)  are replaced by
the Lam\'e-type PWFs.
We will present   details in the next section.

Solutions in Wronskian/Casoratian form allow direct verification
via Hirota bilinear equations. This is known as the  Wronskian/Casoratian technique \cite{FreN-1983}
(see also \cite{Zhang20}).
Let $\mathbf{f}$ be a $N$-th order column vector
\begin{equation}\label{vphi}
\mathbf{f}=(f_1,f_2,\cdots,f_N)^T,
\end{equation}
where $f_j=f_j(x,t)$.
A $N$-th order Wronskian is defined as
\[|~\mathbf{f}, ~ \partial_x \mathbf{f}, ~  \partial_x^2 \mathbf{f},  ~    \cdots,  ~  \partial_x^{N-1}\mathbf{f}|
=|0, ~ 1, ~ 2, ~ \cdots, ~ N-1|=|\widehat{N-1}|,
\]
where we  employ the conventional shorthand introduced in \cite{FreN-1983}.
For convenience we say the above Wronskian is generated by $\mathbf{f}$.
For a given discrete  column vector function $\mathbf{f}=\mathbf{f}(n,m,h)=(f_1(n,m,h),f_2(n,m,h),\cdots,$ $f_N(n,m,h))^T$, the Casorati matrix is constructed of such column vectors with different shifts.
A Casoratian, which is the determinant of the Casorati matrix,
 can be viewed as the discrete analogue of the Wronskian.
 We keep the shorthand notation
\begin{equation*}
|~\mathbf{f}, ~E^{\nu}\mathbf{f}, ~(E^{\nu})^2\mathbf{f}~\cdots,~(E^{\nu})^{N-1}\mathbf{f}|\equiv|\widehat{N-1}|_{[\nu]},
\end{equation*}
where the operators $E^{\nu}(\nu=1,2,3)$ are defined as $E^1\mathbf{f}=\widetilde{\mathbf{f}}=\mathbf{f}(n+1,m,h)$, $E^2\mathbf{f}=\widehat{\mathbf{f}}=\mathbf{f}(n,m+1,h)$, $E^3\mathbf{f}=\overline{\mathbf{f}}=\mathbf{f}(n,m,h+1)$.
The following special Pl\"ucker relation is often used in the Wronskian/Casoratian technique.
\begin{theorem}\label{P-C1}
\cite{FreN-1983} The  relation
	\begin{equation}
	|M,\mathbf{a},\mathbf{b}||M,\mathbf{c},\mathbf{d}|
-|M,\mathbf{a},\mathbf{c}| |M,\mathbf{b},\mathbf{d}|
+|M,\mathbf{a},\mathbf{d}| |M,\mathbf{b},\mathbf{c}|=0
	\end{equation}
holds, where $M$ is a $N\times(N-2)$ matrix, $\mathbf{a}, \mathbf{b}, \mathbf{c}$ and $\mathbf{d}$
are $N$-th order column vectors.
\end{theorem}

\section{Elliptic  \texorpdfstring{$\tau$}{}  functions of the KdV equation and the KP equation}\label{sec-3}

\subsection{\texorpdfstring{$\tau$}{} function of KdV equation}\label{kdv}
In this subsection, we  take the  KdV equation  \eqref{eq:kdv} as a demonstration and  reproduce here some results in \cite{22LxZhang}.
It is notable that, in continuous case, the obtained solutions are no longer elliptic, but still doubly periodic
with respect to some parameters.
We just follow \cite{IMRN2010} and use the term \textit{elliptic  soliton solutions}.

Taking $u=v_x$ we   get the potential KdV equation
 \begin{equation*}
   v_t-\frac{3}{4}v_x^2-\frac{1}{4}v_{xxx}=0,
 \end{equation*}
  which can be bilinearized as
 \begin{equation}\label{bilinear-a}
(D_x^4-4D_xD_t-12\wp(x)D_x^2)f\cdot f=0,
\end{equation}
 by employing the transformation
 \begin{equation}\label{eq:v}
 v =2\zeta(x)+\frac{1}{4}g_2t+2(\ln f)_x.
\end{equation}
On can also introduce $f'=\sigma(x)f$ and
\begin{equation}
	u= 2(\ln f')_{xx}, \quad  v =\frac{1}{4}g_2t+2(\ln f')_x,
\end{equation}
from which the KdV equation \eqref{eq:kdv} is written into an alternative bilinear form
\begin{equation}\label{bilinear-aa}
	(D_x^4-4D_xD_t-g_2)f' \cdot f'=0.
\end{equation}
The bilinear equation \eqref{bilinear-a} admits a Wronskian solution \cite{22LxZhang}
\begin{equation}\label{f-W}
f=|\widehat{N-1}|
\end{equation}
composed by vector $\boldsymbol{\varphi}=(\varphi_1,\varphi_2,\cdots,\varphi_N)^T$.
Each element  $\varphi_j$ is the solution of linear equations
\begin{subequations}
\begin{align}
&\varphi_{j,xx}=(\wp(k_j)+2\wp(x))\varphi_j, \label{phi-jx}\\
& \varphi_{j,t}=\mathbf{\varphi}_{j,xxx}-3\wp(x)\varphi_{j,x}-\frac{3}{2}\wp'(x)\varphi_j, \label{phi-jt}
\end{align}
\end{subequations}
which allow explicit solution in terms of the Weierstrass functions
\begin{equation}
 \varphi_{j}=\varphi_j^++ (-1)^j\varphi_j^-,
\end{equation}
with $ \varphi_j^{\pm}=\varphi^{\pm}(x,t;k_j)$ defined in \eqref{phi-j-gamma}.

To secure the $\tau$ function in Hirota's form, we need to expand the $f$ into $2^N$  distinct
Wronskians, each of which is generated by the elementary column vector of the following form,
\begin{equation}\label{phi-eps}
\boldsymbol{\phi}=(\phi_1, \phi_2, \cdots, \phi_N)^T,~~
\phi_j =(\nu_j)^j\Phi_x(\nu_j k_j)e^{-\nu_j \gamma(k_j)},
\end{equation}
where $\{\nu_1, \nu_2, \cdots, \nu_N\}$ run over $\{1,-1\}$. Recall that $\boldsymbol{\varphi}^-=(\varphi_1^-,\varphi_2^-,\cdots,\varphi_N^-)^T$, and $\varphi_j^-=\varphi^{-}(x,t;k_j)$ is defined in \eqref{phi-j-gamma},
Wronskian composed of $\boldsymbol{\varphi}^-$ can be written as the following.

\begin{lemma}\label{L-1}
For $\boldsymbol{\varphi}^-$ defined as above, we have
\begin{align}
   &|\boldsymbol{\varphi}^-,~ \partial_x \boldsymbol{\varphi}^- ,~\partial_x^2 \boldsymbol{\varphi}^-,~ \cdots  ,~\partial_x^{N-1}\boldsymbol{\varphi}^- |\nonumber \\
=\,&(-1)^N\frac{\sigma(x-\sum_{i=1}^Nk_i)}{\sigma(x)}\cdot
\frac{\prod_{1\leq i<j\leq N}\sigma(k_i-k_j)}{\sigma^{N}(k_1)\cdots\sigma^{N}(k_N)}
\exp\left( \sum_{i=1}^N \gamma(k_i)\right).
\label{W-phi-}
\end{align}
\end{lemma}

 Compared with $\varphi_j^-$, the elementary entries $\phi_j$ defined in \eqref{phi-eps}
 can be formally obtained from $\varphi_j^-$ by multiplying $(\nu_j)^j$
and replacing $k_j$ with $-\nu_j k_j$.
Thus, in light of Lemma \ref{L-1}, the Wronskian generated by \eqref{phi-eps} can be expressed as
\begin{align*}
f_{\boldsymbol{\nu}}\!
= \!(-1)^\frac{N(N-1)}{2}
\frac{\sigma(x+\sum_{i=1}^N \nu_i k_i)}{\sigma(x)}\!
\frac{\prod_{1\leq i<j\leq N}\nu_i\,\sigma(\nu_i k_i-\nu_j k_j)}
{\sigma^{N}(k_1)\cdots\sigma^{N}(k_N)}
\exp\! \left(-\sum_{i=1}^N\nu_i\gamma(k_i)\right)\!,\label{W-phi-eps}
\end{align*}
where $\boldsymbol{\nu}$ indicates cluster
$\boldsymbol{\nu}=\{\nu_1, \nu_2, \cdots, \nu_N\}$.
Introduce length of $\boldsymbol{\nu}$ by $|\boldsymbol{\nu}|$  to denote the number of
positive $\nu_j$'s in the cluster $\boldsymbol{\nu}$.
Rearrange the $2^N$ terms in the $f$ function \eqref{f-W}
in terms of  $|\boldsymbol{\nu}|$ such that
\begin{equation}\label{tau-epl}
f=\sum^{N}_{l=0} \sum_{|\boldsymbol{\nu}|=l}f_{\boldsymbol{\nu}}=\sum^{N}_{l=0}f_{l},
\end{equation}
where $f_{l}=\sum_{|\boldsymbol{\nu}|=l}f_{\boldsymbol{\nu}}$,
 in particular, we have
\begin{align}
f_{0}
=(-1)^{\frac{N(N-1)}{2} }
	\frac{\sigma(x-\sum_{i=1}^Nk_i)}{\sigma(x)}\cdot
\frac{\prod_{1\leq i<j\leq N}\sigma(k_i-k_j)}{\sigma^{N}(k_1)\cdots\sigma^{N}(k_N)}
\exp\left(\sum_{i=1}^N \gamma(k_i)\right). \label{g}
\end{align}

Here, for convenience of this  subsection, for a function $f=f(x)$,
by $f^{\sharp}$ we specially denote the $f$ shifted in $x$ by $\sum_{i=1}^Nk_i$, i.e.
$f^{\sharp}=f(x+\sum_{i=1}^Nk_i)$.
Then we have the following.

\begin{theorem}\label{T-2}
Let
\begin{equation}\label{f-tg}
\tau=\frac{f^{\sharp}}{f_0^{\sharp}},
\end{equation}
where  $f$ and $f_0$ are given by \eqref{f-W} and \eqref{g}.
Then we have
\begin{equation}\label{bilinear-f}
 (D_x^4-4D_xD_t-12\wp(x)D_x^2)\tau\cdot \tau=0,
\end{equation}
the $\tau$ function is in Hirota's form, written as
\begin{equation}\label{f-Hirota}
\tau=\sum_{\mu=0,1} \frac{\sigma(x+2\sum_{i=1}^N \mu_i k_i)}{\sigma(x)\prod^N_{j=1}\sigma^{\mu_j}(2k_j)}
\mathrm{exp}\left(\sum^{N}_{j=1} \mu_j \theta_j+\sum^N_{1\leq i<j}\mu_i\mu_j a_{ij}\right),
\end{equation}
i.e.
\begin{equation*}
\begin{split}
 \tau=1&+ \sum_{i=1}^N\Phi_x(2k_i) e^{\theta_i}
	 +\mathop{\rm{\sum}}_{1\leq l<p\leq N} \frac{\sigma(x+2k_l+2k_p)}{\sigma(x)\sigma(2k_l)\sigma(2k_p)}
A_{lp}e^{\theta_l+\theta_p}\\
	&+\cdots
	 +\frac{\sigma(x+2\sum_{i=1}^Nk_i)}{\sigma(x)\prod^N_{j=1}\sigma(2k_j)}
\left(\prod_{1\leq i<j\leq N}A_{ij}\right)\prod_{i=1}^{N}e^{\theta_i},
    \end{split}
\end{equation*}
where
\begin{subequations}\label{theta-A}
\begin{align}
& \theta_i=-2\zeta(k_i)x + \wp'(k_{i})t +\theta_{i}^{(0)}, ~~ \theta_{i}^{(0)}\in \mathbb{C},\label{theta}\\
& e^{a_{ij}}=A_{ij}=\left(\frac{\sigma(k_i-k_j)}{\sigma(k_i+k_j)}\right)^2,\label{A-ij}
\end{align}
\end{subequations}
and the summation of $\mu$ means to take all possible $\mu_i=\{0,1\}$  for $ i=1,2,\cdots, N$.
\end{theorem}

\begin{remark}
Different from the gauge property presented in Lemma \ref{L-gauge},
when replacing the usual PWF $e^{ax+bt}$ with  the Lam\'e-type PWF $\varrho=\Phi_x (a) e^{bx+ct}$, $a, b, c \in \mathbb{C}$,
we have the quasi-gauge property rather than the gauge property:
\begin{equation}\label{quasi-gauge}
D^n_xD^m_t (\varrho f)\cdot (\varrho g) = \varrho^2  D^n_xD^m_t f\cdot g
+\sum^{\left[\frac{n}{2}\right]}_{l=1}\left(\begin{array}{c}n\\2l\end{array}\right)(D^{2l}_x \varrho\cdot \varrho)
D_x^{n-2l}D^m_t f \cdot g,
\end{equation}
where $\left[\frac{n}{2}\right]$ stands for the floor function of $\frac{n}{2}$,
$f(x,t)$ and $g(x,t)$ are arbitrary $C^{\infty}$ functions.
Noticing that $f^{\sharp}=f_0^{\sharp} \tau$ and $f_0^{\sharp}$ is a Lam\'e-type PWF,
using  the quasi-gauge property, we find
\begin{align*}
     &D_x^4  f^{\sharp} \cdot  f^{\sharp} =
     D_x^4 (f_0^{\sharp} \tau)\cdot (f_0^{\sharp} \tau) =(f_0^{\sharp})^2  D_x^4 \tau\cdot \tau + 6 (D_x^2  f_0^{\sharp}\cdot f_0^{\sharp}) D_x^2 \tau\cdot \tau
     + \tau^2 D_x^4  f_0^{\sharp}\cdot f_0^{\sharp},\\
     &D_x^2 f^{\sharp}\cdot  f^{\sharp} =
     D_x^2 (f_0^{\sharp} \tau)\cdot (f_0^{\sharp} \tau) =(f_0^{\sharp})^2  D_x^2 \tau\cdot \tau+  \tau^2 D_x^2  f_0^{\sharp}\cdot f_0^{\sharp},\\
	 &D_xD_t f^{\sharp}\cdot  f^{\sharp}=
D_x D_t (f_0^{\sharp} \tau)\cdot (f_0^{\sharp} \tau) =(f_0^{\sharp})^2  D_xD_t \tau\cdot \tau,
\end{align*}
and meanwhile, the term $D_x^{2l}~ f_0^{\sharp}\cdot f_0^{\sharp}$ can be replaced by
\begin{align*}
     &D_x^2~ f_0^{\sharp}\cdot f_0^{\sharp}=2\left(\wp(x+\hbox{$\sum_{i=1}^Nk_i$})-\wp(x)\right) (f_0^{\sharp})^2,\\
	 &D_x^4~ f_0^{\sharp}\cdot f_0^{\sharp}=12\wp(x+\hbox{$\sum_{i=1}^Nk_i$}) D_x^2~ f_0^{\sharp}\cdot f_0^{\sharp}.
\end{align*}
These results give rise to
\begin{equation*}
\begin{split}
0=&\, (D_x^4-4D_xD_t-12\wp(\hbox{$x+\sum_{i=1}^Nk_i)D_x^2$})f^{\sharp}\cdot f^{\sharp}\\
=&\, (f_0^{\sharp})^2 (D_x^4-4D_xD_t-12\wp(x)D_x^2)\tau\cdot \tau,
\end{split}
\end{equation*}
which indicates that $\tau=f^{\sharp} /f_0^{\sharp}$ solves the bilinear KdV equation \eqref{bilinear-f}.
\end{remark}

By introducing infinite coordinates, the above $\tau$  function can be generalized to the KdV hierarchy. Let us first list some notations:
\begin{subequations}\label{xi-theta}
	\begin{align}
    &\mathbf{t}=(t_1=x,t_2, \cdots, t_n, \cdots), ~~~\overline{\mathbf{t}}=(t_1=x,t_3, \cdots, t_{2n+1}, \cdots),\\
    &\widetilde{\partial}=(\partial_{t_1},\frac{1}{2} \partial_{t_2},\cdots,\frac{1}{n} \partial_{t_n},\cdots),~~~
    \overline{\widetilde{\partial}}=(\partial_{t_1},\frac{1}{3} \partial_{t_3},\cdots,\frac{1}{(2n+1)} \partial_{t_{2n+1}},\cdots),\\
	&\xi (\mathbf{t},k)=\sum_{n=1}^{\infty}k^n t_n, ~~~
      \xi_{[e]}(\mathbf{t},k)=\sum_{n=1}^{\infty}(-1)^n\frac{\zeta^{(n-1)}(k)}{(n-1)!}t_n,~~
      \zeta^{(i)}(k)=\partial^i_k\zeta(k),\\
    & \theta(\overline{\mathbf{t}},k)=\xi (\mathbf{t},k)-\xi (\mathbf{t},-k)
    =2\sum_{n=0}^{\infty} k^{2n+1} t_{2n+1},\\
    & \theta_{[e]}(\overline{\mathbf{t}},k)=\xi_{[e]} (\mathbf{t},k)-\xi_{[e]} (\mathbf{t},-k)
    =-2\sum_{n=0}^{\infty}\frac{\zeta^{(2n)}(k)}{(2n)!} t_{2n+1}.
\end{align}
\end{subequations}
Consider the following $\tau$ function which is equivalent to \eqref{f-Hirota},
\begin{equation}\label{tau-ver}
\tau_N^{}(\overline{\mathbf{t}})=\sum_{J\subset S}\left(\prod_{i\in J}c_i\right)\Biggl(\prod_{i,j\in J\atop i<j}A_{ij}\Biggr)
\frac{\sigma(t_1+2\sum_{i\in J}k_i)}{\sigma(t_1)\prod_{i\in J}\sigma(2k_i)}
\exp\left(\sum_{i\in J}\theta_{[e]}(\overline{\mathbf{t}},k_i)\right),
\end{equation}
where $c_i$ are arbitrary constants, $A_{ij}$ is defined as in \eqref{A-ij},
$S=\{1,2,\cdots,N\}$, $J$ is a subset of $S$,
and $\sum_{J\subset S}$ means the summation runs over all subsets of $S$.
Considering the Taylor series of $A_{ij}$ in the neighborhood of $q=0$, we have
\begin{equation}\label{Aij-pq}
\ln \frac{\sigma(p-q)}{\sigma(p+q)}=\theta_{[e]}(\bar\varepsilon(q), p),
\end{equation}
which indicates
\begin{equation}
A_{ij}=e^{2\theta_{[e]}(\bar\varepsilon(k_j), k_i)},
\end{equation}
where $\bar\varepsilon(q)=(q,\frac{q^3}{3},\cdots,\frac{q^{2n+1}}{(2n+1)},\cdots)$.
Furthermore,
for $\theta$ and $\theta_{[e]}$ defined in \eqref{xi-theta}, we have
\begin{equation*}\label{X-exch}
e^{\theta (\overline{\widetilde{\partial}},k_i)} e^{\theta_{[e]} (\bar{\mathbf{t}},k_j)}=A_{ij} \,
e^{\theta_{[e]} (\bar{\mathbf{t}},k_j)} e^{\theta (\overline{\widetilde{\partial}},k_i)},
\end{equation*}
where
$A_{ij}$ is defined as \eqref{A-ij}.

The vertex operator that generates the above $\tau$ function is described below.
\begin{theorem}\label{T-3}
The $\tau$ function \eqref{tau-ver} can be generated by the vertex operator
\begin{equation}\label{eq:vertexkdv}
X(k)=\Phi_{t_1}(2 k) e^{\theta_{[e]} (\overline{\mathbf{t}},k)} e^{\theta (\overline{\widetilde \partial},k)}
\end{equation}
via
\begin{equation}\label{tau-N-N}
 \tau_{N}(\overline{\mathbf{t}})=e^{c_{N}X(k_{N})}\,\tau_{N-1}(\overline{\mathbf{t}}),~~
 \tau_{0}(\overline{\mathbf{t}})=1,
\end{equation}
i.e.
\begin{equation}\label{tau-N}
\tau_N(\overline{\mathbf{t}})=e^{c_{N}X(k_N)}\cdots e^{c_{2}X(k_2)}e^{c_{1}X(k_1)}\, 1.
\end{equation}
\end{theorem}
The above theorem can be directly verified by the commutative relations between the vertex operators,
such as
\begin{subequations}
\begin{align*}
	&X(k_i)X(k_j)=A_{ij}\frac{\sigma(t_1+2k_i+2 k_j)}{\sigma(t_1)\sigma(2k_i)\sigma(2k_j)}
      e^{\theta_{[e]}(\bar{\mathbf{t}},k_i)+\theta_{[e]}(\bar{\mathbf{t}},k_j)}
      e^{\theta(\overline{\widetilde{\partial}}, k_i)+\theta(\overline{\widetilde{ \partial}}, k_j)},\label{XX-commu}\\
	&X(k_s)\cdots X(k_2)X(k_1) \nonumber \\
	&\qquad=\left(\prod_{1\leq i<j\leq s}A_{ij}\right)
       \frac{\sigma(t_1+2\sum_{i=1}^sk_i)}{\sigma(t_1)\prod_{i=1}^{s}\sigma(2k_i)}
       \exp\left( \sum_{i=1}^s\theta_{[e]}(\overline{\mathbf{t}},k_i)\right)
       \exp\left(\sum_{i=1}^s\theta(\overline{\widetilde{\partial}}, k_i)\right),
\end{align*}
\end{subequations}
and hence
\begin{equation*}
X(k)^2=0,~~
e^{cX(k)}=1+cX(k).
\end{equation*}

The vertex operators can be viewed here as the  B\"acklund transformation to generate new solutions. In addition, noticing that $\tau_1(\bar{\mathbf{t}})= e^{c_{1}X(k_1)}\, 1$ is doubly periodic with respect to $k_1$,
and $X(k_i)$ and $X(k_j)$ commute,
it  follows that $\tau_N(\bar{\mathbf{t}})$ defined by \eqref{tau-N} is
doubly periodic with respect to any $k_i$, for $i=1,2, \cdots, N$.

\subsection{\texorpdfstring{$\tau$}{} function of KP equation}\label{kp}

Let us partially review the results about the KP hierarchy. Details can be found in \cite{22LxZhang}.
The KP equation is
\begin{equation}\label{eq:kp}
	4u_t-u_{xxx}-6uu_x-3\partial^{-1}u_{yy}=0,
\end{equation}
or in the potential form $(u=v_x)$
\begin{equation}\label{eq:kp-v}
	4v_t-v_{xxx}-3(v_x)^2-3\partial^{-1}v_{yy}=0.
\end{equation}
By the transformation
\begin{align}
u=-2\wp(x)+2(\ln f)_{xx},
\end{align}
or
\begin{equation}
v=2\zeta(x)+\frac{g_2}{4} t+2(\ln f)_x,
\end{equation}
the KP equation is bilinearised as
\begin{equation}\label{bilinear-b}
(D_x^4- 4D_xD_t-12\wp(x)D_x^2+3D_y^2)f\cdot f=0,
\end{equation}
or
\begin{equation}
	(D_x^4-4D_xD_t+3 D_y^2-g_2)f'\cdot f'=0,
\end{equation}
where $f'=\sigma(x)f$. The bilinear KP equation allows elliptic soliton solutions.
\begin{theorem}\label{T-6}
The following Wronskian
\begin{equation}\label{slt:KPWronskian}
	f=|\widehat{N-1}|
\end{equation}
is a solution to the bilinear KP equation \eqref{bilinear-b}, where $f$ is composed by vector
$\boldsymbol{\varphi}\!=\!(\varphi_1, \cdots, \varphi_N)^T$
with entries
\begin{subequations}
\begin{equation}
\varphi_j(x,y,t) = \Phi_x(k_j)e^{-\gamma(k_j)}+ \Phi_x(l_j)e^{-\gamma(l_j)},
\end{equation}
where
\begin{equation}
\gamma(k)=\zeta(k)x+ \wp(k)y-\frac{\wp'(k)}{2}t+\gamma^{(0)}(k),~~ k\in \mathbb{C}
\end{equation}
with a constant $\gamma^{(0)}(k)$ related to $k$.
\end{subequations}
Note that $\varphi_j$ satisfies
\begin{equation}\label{eq:kplax}
	\begin{split}
	&\varphi_{j,y}= -\varphi_{j,xx}+2\wp(x)\varphi_j,\\
	&\varphi_{j,t}=\varphi_{j,xxx}-3\wp(x)\varphi_{j,x}-\frac{3}{2}\wp'(x)\varphi_j.
	\end{split}
\end{equation}
\end{theorem}

To find out a corresponding Hirota's form of the $f$ function \eqref{slt:KPWronskian},
we consider \eqref{slt:KPWronskian} to be a summation of $2^N$ terms,
i.e. $f=\sum_{J \subset S}f_{J}^{}$,
where the generic term $f_{J}^{}$ is the Wronskian $|\widehat{N-1}|$ generated by
\begin{equation}\label{vphi-J}
\boldsymbol{\phi}=(\phi_1, \phi_2, \cdots, \phi_N)^T,
\end{equation}
in which $\phi_j=\Phi_x(k_j)e^{-\gamma(k_j)}$ for $j\in J$
and $\phi_j=\Phi_x(l_j) e^{-\gamma(l_j)}$ for $j\in S\backslash J$,
$J$ is a subset of $S=\{1,2,\cdots,N\}$.
In light of Lemma \ref{L-1}, one can  get the following result.

\begin{lemma}\label{L-5}
The Wronskian $f_{J}^{}$ generated by vector \eqref{vphi-J} can be expressed as
\begin{align}
f_{J}^{}=&(-1)^{\frac{N(N-1)}{2}}
\frac{\sigma(x+\sum_{i\in J }k_{i}+\sum_{j\in S\backslash J}l_j)}{\sigma(x)}
\frac{\prod_{i\in J \atop j\in S\backslash J}\sigma(k_i-l_j)\mathrm{sgn}[j-i]}
{\left(\prod_{i\in J}\sigma^N(k_{i})\right)\left(\prod_{ j\in S\backslash J}\sigma^N(l_{j})\right)} \nonumber\\
 & \times \left(\prod_{i<j\in J}\sigma(k_i-k_j) \right)
 \left(\prod_{i<j\in S\backslash J}\sigma(l_{i}-l_j)\right)
 \exp\left[-\sum_{i\in J}\gamma(k_i)-\sum_{j\in S\backslash J}\gamma(l_j)\right],
\end{align}
especially, when $J$ is the empty set $\varnothing$, we have
\begin{equation}\label{g-KP}
f_{\varnothing}^{}
=(-1)^{\frac{N(N-1)}{2}}
\frac{\sigma(x+\sum_{j\in S}l_j)}{\sigma(x)}
\frac{\prod_{i<j \in S}\sigma(l_i-l_j)}
{\prod_{ j\in S}\sigma^N(l_{j})} \mathrm{exp}\left(-\sum_{j\in S}\gamma(l_j)\right).
\end{equation}
\end{lemma}

Next, for a function $f(x)$, setting notation $f^{\natural}=f(x-\sum_{j=1}^N l_j)$,
and similar to the KdV case, we have the following.
\begin{theorem}\label{T-7}
For the function $f$ in Wronskian form \eqref{slt:KPWronskian} and $f_{\varnothing}$  given by \eqref{g-KP},
\begin{equation}\label{f-tg-KP}
\tau=\frac{f^{\natural}}{f_{\varnothing}^{\natural}}
\end{equation}
is a solution to the bilinear KP equation \eqref{bilinear-b}, i.e.
\begin{equation}\label{bilinear-bf}
(D_x^4- 4D_xD_t-12\wp(x)D_x^2+3D_y^2)\tau \cdot \tau=0,
\end{equation}
and $\tau$ is written in Hirota's form as
\begin{equation}\label{f-Hirota-KP}
\tau=\sum_{\mu=0,1} \frac{\sigma(x+\sum_{i=1}^N \mu_i (k_i-l_i))}{\sigma(x)\prod^N_{i=1}\sigma^{\mu_i}(k_i-l_i)}
\mathrm{exp}\left(\sum^{N}_{j=1} \mu_j \theta_j+\sum^N_{1\leq i<j}\mu_i\mu_j a_{ij}\right),
\end{equation}
where
\begin{subequations}\label{theta-A-KP}
\begin{align}
& \theta_i=-(\zeta(k_i)-\zeta(l_i)) x -(\wp(k_i)-\wp(l_i)) y +\frac{1}{2}(\wp'(k_i)-\wp'(l_i))t +\theta_{i}^{(0)}, ~
\theta_{i}^{(0)}\in \mathbb{C},\label{theta-KP}\\
& e^{a_{ij}}=A_{ij}=\frac{\sigma(k_i-k_j)\sigma(l_i-l_j)}{\sigma(k_i-l_j)\sigma(l_i-k_j)},\label{A-ij-KP}
\end{align}
\end{subequations}
and the summation of $\mu$ means to take all possible $\mu_i=\{0,1\}$  for $ i=1,2,\cdots, N$.
\end{theorem}
We now present a vertex operator
\begin{equation}\label{eq:vertexkp}
X(k,l)=\Phi_{t_1}(k-l)e^{\xi_{[e]}(\mathbf{t},k)-\xi_{[e]}(\mathbf{t},l)}
e^{\xi(\widetilde \partial,k)-\xi(\widetilde \partial,l)},
\end{equation}
where $\widetilde{\partial}$, $\xi$ and $\xi_{[e]}$ are defined in \eqref{xi-theta}.
 The $\tau$ functions for
elliptic soliton solutions of the KP hierarchy can be generated as following.
\begin{theorem}\label{T-8}
For the KP hierarchy, its $\tau$ function of elliptic $N$-soliton solution,
\begin{equation}\label{tau-KP}
\tau_N(\mathbf{t})=\sum_{J\subset S}\left(\prod_{i\in J}c_i\right)
\left(\mathop{\rm{\prod}}_{i<j \in J}A_{ij}\right)
\frac{\sigma(t_1+\sum_{i\in J}(k_i-l_i))}{\sigma(t_1)\prod_{i\in J}\sigma(k_i-l_i)}
e^{\sum_{i\in J}(\xi_{[e]}(\mathbf{t}, k_i)-\xi_{[e]}(\mathbf{t},l_i))},
\end{equation}
is generated by the vertex operator \eqref{eq:vertexkp} via
\begin{equation}
\tau_N(\mathbf{t})=e^{c_{N}X(k_N,l_N)}\cdots e^{c_{2}X(k_2,l_2)}e^{c_{1}X(k_1,l_1)}\,  1,
\end{equation}
or via transformation
\begin{equation}
\tau_{N}(\mathbf{t})=e^{c_{N}X(k_{N}, l_{N})}\,  \tau_{N-1}(\mathbf{t}),~ ~ \tau_{0}(\mathbf{t})=1.
\end{equation}
In addition, $\tau_N(\mathbf{t})$ is a doubly periodic function with respect to any $k_i$ and $l_j$
for $i,j\!=1,2, \cdots, N$.
\end{theorem}

\section{Elliptic \texorpdfstring{$\tau$}{} function of lattice potential KdV equation}\label{lkdv}

In this section we come to the discrete case and derive elliptic soliton solution
for the lattice potential KdV (lpKdV) equation.
The solutions will be  presented in  Hirota's form and proved via bilinear lpKdV equations.
The usual soliton solution  of this equation  has been obtained in \cite{09HZ}.

We employ the conventional shorthand notations for the shifts of a function $u(n,m,h)$,  $n,m,h\in \mathbb{Z}$:
\begin{equation}
 u:=u(n,m,h),~~ \tilde{u}:=u(n+1,m,h),~~  \hat{u}:=u(n,m+1,h),~~ \bar{u}:=u(n,m,h+1),
\end{equation}
and the shifts in the opposite direction are denoted by under-tildes, under-hats and under-bars, like $\undertilde{u}:=u(n-1,m,h)$.

The lpKdV equation \cite{83-NQC}
\begin{equation}\label{eq:lpkdv}
(w-\widehat{\widetilde{w}})(\widehat{w}-\widetilde{w})
=a-b,
\end{equation}
is a well-known integrable discrete equation with 3-dimensional consistency,
which is also the H1 equation in the Adler-Bobenko-Suris classification
for affine quadrilateral equation that are consistent around the cube \cite{ABS-2003}.
The lpKdV equation has a Lax pair \cite{HJN-book}
\begin{subequations}\label{eq:lax-h1}
\begin{align}
&\widetilde{\widetilde{\varphi}}+(\widetilde{\widetilde{w}}-w)\widetilde{\varphi}=(\lambda-a)\varphi,\\
& \widehat{\varphi}=\widetilde{\varphi}-(\widehat{w}-\widetilde{w})\varphi.
\end{align}
\end{subequations}
With new  parametrizations
\begin{equation}\label{pq-para}
a=\wp(\delta)-e_1,~~ b=\wp(\varepsilon)-e_1,
\end{equation}
 the lpKdV  equation  allows a simple solution
 \begin{equation}\label{slt:w0}
w_0 =\zeta(\xi)-n\zeta(\delta)-m\zeta(\varepsilon)-\zeta(\xi_0), \qquad \xi=n\delta + m \varepsilon,
\end{equation}
where $e_1,   \xi_0 \in \mathbb{C}$,
and $\delta, \varepsilon$ are lattice parameters.

To obtain the $\tau$-function in Hirota's form, we derive solution in Casoratian form using bilinear approach.
We introduce the elementary column  vector
$\boldsymbol{\varphi}=(\varphi_{1}(n,m,h),\cdots, \varphi_{N}(n,m,h))^T$
in  Casoratians, where
\begin{subequations}\label{phi-h1}
  \begin{equation}
  \begin{split}
\varphi_{i}&=\varphi_i^+ +(-1)^i\varphi_i^-\\
&=\rho_{n,m,h}^-(k_i)\Phi_{\xi}(k_i)+(-1)^i\rho_{n,m,h}^+(k_i)\Phi_{\xi}(-k_i),
\end{split}
\end{equation}
with
\begin{equation}\label{rho+-}
\rho_{n,m,h}^{\pm}(z)=\Big(\Phi_{\delta}(\pm z)\Big)^n\Big(\Phi_{\varepsilon}(\pm z)\Big)^m
\Big(\Phi_{\gamma}(\pm z)\Big)^h,
\end{equation}
\end{subequations}
and $\Phi_z(u)$ defined in \eqref{Phi}. Here the last index $h$ indicates the shift in Casoratian. Each $\varphi_i$ is the solution of Eq.\eqref{eq:lax-h1} with $\lambda=\wp(k_i)-e_1$ and $w=w_0$.
Due to the relations
\begin{subequations}\label{eq:varphi-shift-1}
\begin{align}
\widehat{\boldsymbol{\varphi}}&
                 =\widetilde{\boldsymbol{\varphi}}+\chi_{\varepsilon,-\delta}(\xi+\delta)\boldsymbol{\varphi},\\
\overline{\boldsymbol{\varphi}}&
                  =\widetilde{\boldsymbol{\varphi}}+\chi_{\gamma,-\delta}(\xi+\delta)\boldsymbol{\varphi},\\
\overline{\boldsymbol{\varphi}}&
                  =\widehat{\boldsymbol{\varphi}}+\chi_{\gamma,-\varepsilon}(\xi+\varepsilon)\boldsymbol{\varphi},
\end{align}
\end{subequations}
we can easily verify that
\begin{equation}
|\widehat{N-1}|_{[1]}=|\widehat{N-1}|_{[2]}=|\widehat{N-1}|_{[3]}.
\end{equation}
For the elliptic $N$-soliton solution of the lpKdV equation, we have the following.

\begin{lemma}\label{lemma-1}
 The lpKdV equation \eqref{eq:lpkdv} with parametrizations \eqref{pq-para}
admits elliptic $N$-soliton solution in Casoratian form:
 \begin{equation}\label{seed-H1}
\begin{split}
 w&=\zeta(\xi+N\gamma)-N\zeta(\gamma)-n\zeta(\delta)-
 m\zeta(\varepsilon)-h\zeta(\gamma)-\zeta(\xi_0)+\frac{g}{f}
\end{split}
\end{equation}
 where $\xi=n\delta+m\varepsilon+h\gamma$,
\begin{align}\label{h1-fg}
f=f(\varphi)=\sigma(\xi)|\widehat{N-1}|, \qquad  g=g(\varphi)=\sigma(\xi)|\widehat{N-2},~N| ,
\end{align}
in which the basic elements $\varphi_i $ are defined as in \eqref{phi-h1}.
 \end{lemma}

\begin{proof}
First, the lpKdV equation can be bilinearized as the following
\begin{subequations}\label{4.11}
\begin{align}
\mathcal{H}_1&\equiv \chi_{\delta,-\varepsilon}(\widehat{\xi}+N\gamma)\widetilde{f}\widehat{f}+\widetilde{f}\widehat{g}-\widetilde{g}\widehat{f}
-\Phi_{\delta}(-\varepsilon)f\widetilde{\widehat{f}}=0,
\label{bilinear-2}\\
\mathcal{H}_2&\equiv
\chi_{\delta,\varepsilon}(\xi+N\gamma)f\widetilde{\widehat{f}}+\widehat{\widetilde{f}}g-\widehat{\widetilde{g}}f
-\Phi_{\delta}(\varepsilon)\widetilde{f}\widehat{f}=0,\label{bilinear-3}
\end{align}
\end{subequations}
where $\Phi_z(u)$ and $\chi_{u,v}(z)$ are defined as in \eqref{Phi} and \eqref{eq:add-4}, respectively.
Viewing \eqref{seed-H1} as a  transformation, it is easy to check that
\begin{equation}
\begin{split}
&  (w-\widehat{\widetilde{w}})(\widehat{w}-\widetilde{w})-\wp(\delta)+\wp(\varepsilon)\\
 =  &\frac{1}{f\widetilde{f}\widehat{f}\widetilde{\widehat{f}}}
 \Bigl(\mathcal{H}_1+\Phi_{\delta}(-\varepsilon)f\widetilde{\widehat{f}}\Bigr)
 \left(\mathcal{H}_2+\Phi_{\delta}(\varepsilon)\widetilde{f}\widehat{f}\right)
-\wp(\delta)+\wp(\varepsilon),
\end{split}
\end{equation}
where we have made use of the formula \eqref{eq:add-1}.

Next, we prove the bilinear equations \eqref{4.11} allow Casoratian solution given in \eqref{h1-fg}.
We prove \eqref{bilinear-2} in its down-tilde-hat-shifted version
\begin{equation}\label{bilinear-2-1}
\chi_{\delta,-\varepsilon}(\undertilde{\xi}+N\gamma)\undertilde{f}\underset{\widehat{}}{f}
+\undertilde{g}\underset{\widehat{}}{f}
-\undertilde{f}\underset{\widehat{}}{g}
-\Phi_{\delta}(-\varepsilon)
f\underset{\widehat{}}{\undertilde{f}}=0.
\end{equation}
In light of the relation \eqref{eq:varphi-shift-1}, and let $(E^3)^j\boldsymbol{\varphi}=\boldsymbol{\varphi}(n,m,h+j\gamma)=\boldsymbol{\varphi}(j)$, we immediately get
\begin{equation}\label{eq:varphi-shift-2}
    \undertilde{\boldsymbol{\varphi}}(j+1)=\boldsymbol{\varphi}(j)+\chi_{\gamma,-\delta}(\xi+j\gamma)\undertilde{\boldsymbol{\varphi}}(j), \quad
    \underset{\widehat{}}{\boldsymbol{\varphi}}(j+1)
    =\boldsymbol{\varphi}(j)+\chi_{\gamma,-\varepsilon}(\xi+j\gamma)\underset{\widehat{}}{\boldsymbol{\varphi}}(j). \\
\end{equation}
Therefore $f$ and $g$ satisfy the following relations,
\begin{subequations}
\begin{align*}
\undertilde{f}&
             =(-1)^{N-1}\sigma(\undertilde{\xi})\frac{\left|\widehat{N-2},  ~ \undertilde{\boldsymbol{\varphi}} {(N-2)}\right|}{X(\gamma,-\delta,N-3)},\\
\underset{\widehat{}}{f}&
           =(-1)^{N-1}\sigma(\underset{\widehat{}}{\xi})\frac{\left|\widehat{N-2}, ~ \underset{\widehat{}}{\boldsymbol{\varphi}}(N-2)\right|}
           {X(\gamma,-\varepsilon,N-3)},\\
 \underset{\widehat{}}{\undertilde{f}}&=\sigma( \underset{\widehat{}}{\undertilde{\xi}})\frac{\left|\widehat{N-3}, ~
           \underset{\widehat{}}{\boldsymbol{\varphi}}(N-2) ,  ~ \underset{\widehat{}}{\undertilde{\boldsymbol{\varphi}}}(N-2)\right|}
           {-\underset{\widehat{}}{X}(\gamma,-\delta,N-3) X(\gamma,-\varepsilon,N-3)},\\
\undertilde{g}&=\sigma(\undertilde{\xi})\frac{\left|\widehat{N-3}, ~ \undertilde{\boldsymbol{\varphi}}(N-2), ~ \boldsymbol{\varphi}(N-1)\right|
            -\chi_{\gamma,-\delta}(\xi+(N-1)\gamma)\left|\widehat{N-2}, ~ \undertilde{\boldsymbol{\varphi}}(N-2)\right|}
           {(-1)^{N-2}X(\gamma,-\delta,N-3)},\\
\underset{\widehat{}}{g}&=\sigma(\underset{\widehat{}}{\xi})\frac{\left|\widehat{N-3}, ~ \underset{\widehat{}}{\boldsymbol{\varphi}}(N-2), ~
            \boldsymbol{\varphi}(N-1)\right|-\chi_{\gamma,-\varepsilon}(\xi+(N-1)\gamma)\left|\widehat{N-2},  ~ \underset{\widehat{}}{\boldsymbol{\varphi}}(N-2)\right|} {(-1)^{N-2}X(\gamma,-\varepsilon,N-3)},
\end{align*}
\end{subequations}
where
\begin{equation}\label{X}
X(a,b,N)=\prod_{j=0}^{N}\chi_{a,b}(\xi+j a).
\end{equation}
Substituting these relations into \eqref{bilinear-2-1} and multiplying  $X(\gamma,-\delta,N-3)X(\gamma,-\varepsilon,N-3)$,
the left-hand side of \eqref{bilinear-2-1} turns out to be
\begin{equation*}
\begin{split}
&|\widehat{N-3},  ~ \underset{\widehat{}}{\boldsymbol{\varphi}}(N-2), ~ \boldsymbol{\varphi}(N-1)|   |\widehat{N-2},  ~ \undertilde{\boldsymbol{\varphi}}(N-2)| \\
- &|\widehat{N-3}, ~ \undertilde{\boldsymbol{\varphi}}(N-2), ~ \boldsymbol{\varphi}(N-1)|  |\widehat{N-2},  ~ \underset{\widehat{}}{\boldsymbol{\varphi}}(N-2)|\\
+&\frac{\sigma(\xi)\sigma(\xi-\delta-\varepsilon)}{\sigma(\xi-\delta)\sigma(\xi-\varepsilon)}
\frac{X(\gamma,-\delta,N-3)}{ \underset{\widehat{}}{X}(\gamma,-\delta,N-3)} \Phi_{\delta}(-\varepsilon)
|\widehat{N-1}| |\widehat{N-3},  ~\underset{\widehat{}}{\boldsymbol{\varphi}}(N-2), ~ \underset{\widehat{}}{\undertilde{\boldsymbol{\varphi}}}(N-2)|.
\end{split}
\end{equation*}
Using the definition of $\chi_{u,v}(z)$ \eqref{eq:add-4} and the following relation
\begin{equation*}
\left(\chi_{\gamma,-\delta}(\underset{\widehat{}}{\xi}+(N-2)\gamma)-\chi_{\gamma,-\varepsilon}(\undertilde{\xi}+(N-2)\gamma)\right)
\underset{\widehat{}}{\undertilde{\boldsymbol{\varphi}}}(N-2)=\undertilde{\boldsymbol{\varphi}}(N-2)-\underset{\widehat{}}{\boldsymbol{\varphi}}(N-2),
\end{equation*}
we can get
\begin{equation*}
\begin{split}
&|\widehat{N-3}, ~\underset{\widehat{}}{\boldsymbol{\varphi}}(N-2), ~\boldsymbol{\varphi}(N-1)| |\widehat{N-2}, ~ \undertilde{\boldsymbol{\varphi}}(N-2)|\\
-&|\widehat{N-3}, ~\undertilde{\boldsymbol{\varphi}}(N-2), ~\boldsymbol{\varphi}(N-1)| |\widehat{N-2},  ~\underset{\widehat{}}{\boldsymbol{\varphi}}(N-2)|
 \\
-&|\widehat{N-1} | |\widehat{N-3}, ~ \underset{\widehat{}}{\boldsymbol{\varphi}}(N-2), ~ \undertilde{\boldsymbol{\varphi}}(N-2)|,
\end{split}
\end{equation*}
which is zero in light of Theorem \ref{P-C1}.
Thus \eqref{bilinear-2} is proved.

Next we prove \eqref{bilinear-3} in its down-tilde-shifted version
\begin{equation}\label{bilinear-3-1}
\chi_{\delta,\varepsilon}(\undertilde{\xi}+N\gamma)\undertilde{f}\widehat{f}+\widehat{f}\undertilde{g}-\widehat{g}\undertilde{f}
-\Phi_{\delta}(\varepsilon)f\undertilde{\widehat{f}}=0.
\end{equation}
To achieve that, we introduce the auxiliary function
\begin{equation}
\begin{split}
\psi_i&=\varphi_i\times\Big(\wp(k_i)-\wp(\varepsilon)\Big)^{-m}\\
&=\Big(\Phi_{\delta}(k_i)\Big)^n\Big(-\Phi_{\varepsilon}(-k_i)\Big)^{-m}
\Big(\Phi_{\gamma}(k_i)\Big)^h\Phi_{\xi}(-k_i)\\
&+\Big(\Phi_{\delta}(-k_i)\Big)^n\Big(-\Phi_{\varepsilon}(k_i)\Big)^{-m}
\Big(\Phi_{\gamma}(-k_i)\Big)^h\Phi_{\xi}(k_i),
\end{split}
\end{equation}
which satisfies
\begin{equation}
\begin{split}
\widehat{\overline{\psi}}_i=\psi_i+\chi_{\gamma,\varepsilon}(\xi)\widehat{\psi_i}, \qquad
\overline{\undertilde{\psi}}{}_i=\psi_i+\chi_{\gamma,-\delta}(\xi)\undertilde{\psi_i}.
\end{split}
\end{equation}
We denote the Casoratians composed by $\boldsymbol{\psi}$ and its shifts as
$\bigl|\widehat{N-1}\bigr|_{\boldsymbol{\psi}},~\bigl|\widehat{N-2},N\bigr|_{\boldsymbol{\psi}}$, etc,
which are distinguished by the index $\boldsymbol{\psi}$.
Such Casoratians are connected with $f(\boldsymbol{\varphi})$ and $g(\boldsymbol{\varphi})$ defined in \eqref{h1-fg}
as the following.
\begin{subequations}
\begin{align*}
f(\boldsymbol{\varphi}) &=\Gamma^m\sigma(\xi)\bigl|\widehat{N-1}\bigr|_{\boldsymbol{\psi}},\\
\undertilde{f}(\boldsymbol{\varphi})&=\Gamma^m\sigma(\undertilde{\xi})
       \frac{\bigl|\widehat{N-2},  ~\undertilde{\boldsymbol{\psi}}(N-2)\bigr|}{(-1)^{N-1}X(\gamma,-\delta,N-3)},\\
\widehat{f}(\boldsymbol{\varphi})&=\Gamma^{m+1}\sigma(\widehat{\xi})
       \frac{\big|\widehat{N-2}, ~ \widehat{\boldsymbol{\psi}}(N-2)\bigr|}
       {(-1)^{N-1}X(\gamma,\varepsilon,N-3)},\\
\undertilde{\widehat{f}}(\boldsymbol{\varphi})&=\Gamma^{m+1}\sigma(\undertilde{\widehat{\xi}})
       \frac{\left| \widehat{N-3}, ~  \widehat{\boldsymbol{\psi}}(N-2),  ~ \undertilde{\widehat{\boldsymbol{\psi}}}(N-1)\right|}  {-X(\gamma,\varepsilon,N-3)X(\gamma,-\delta,N-2)},\\
\undertilde{g}(\boldsymbol{\varphi})&=\Gamma^{m}\sigma(\undertilde{\xi})\\
       &\times \frac{\Bigl|\widehat{N-3}, ~ \undertilde{\boldsymbol{\psi}}(N-2), ~ \boldsymbol{\psi}(N-1)\Bigr|-\chi_{\gamma,-\delta}(\xi+(N-1)\gamma)
       \left|\widehat{N-2}, ~ \undertilde{\boldsymbol{\psi}}(N-2)\right|}
       {(-1)^{N-2}X(\gamma,-\delta,N-3)},\\
\widehat{g}(\boldsymbol{\varphi})&=\Gamma^{m+1}\sigma(\widehat{\xi})\\
       &\times\frac{\bigl| \widehat{N-3}, ~ \widehat{\boldsymbol{\psi}}(N-2), ~ \boldsymbol{\psi}(N-1)\bigr| - \chi_{\gamma,\varepsilon}(\xi+(N-1)\gamma)\left|\widehat{N-2}, ~ \widehat{\boldsymbol{\psi}}(N-2)\right|}
       {(-1)^{N-2}X(\gamma,\varepsilon,N-3)},
\end{align*}
\end{subequations}
where $\Gamma=\prod_{i=1}^N\big(\wp(k_i)-\wp(\varepsilon)\big)$.
Substituting these expressions into \eqref{bilinear-3-1} and multiplying
$\Gamma^{2m+1}X(\gamma,\varepsilon,N-3)X(\gamma,-\delta,N-3)$,
one  obtains
\begin{equation*}
\begin{split}
&| \widehat{N-3}, ~ \widehat{\boldsymbol{\psi}}(N-2), ~ \boldsymbol{\psi}(N-1)| |\widehat{N-2},  ~ \undertilde{\boldsymbol{\psi}}(N-2)|\\
-&| \widehat{N-3}, ~ \undertilde{\boldsymbol{\psi}}(N-2), ~ \boldsymbol{\psi}(N-1)| |\widehat{N-2},  ~ \widehat{\boldsymbol{\psi}}(N-2)|\\
+&\frac{\sigma(\xi)\sigma(\xi+\varepsilon-\delta)}{\sigma(\xi-\delta)\sigma(\xi+\varepsilon)}
\frac{X(\gamma,-\delta,N-3)}{\widehat{X}(\gamma,-\delta,N-3)} \Phi_{\delta}(\varepsilon)
|\widehat{N-1}|_{\boldsymbol{\psi}}\frac{\left| \widehat{N-3}, ~ \widehat{\boldsymbol{\psi}}(N-2), ~ \undertilde{\widehat{\boldsymbol{\psi}}}(N-1)\right|}
{\chi_{\gamma,-\delta}(\widehat{\xi}+(N-2)\gamma)}.
\end{split}
\end{equation*}
With the help of
\begin{equation*}
\begin{split}
&\frac{\chi_{\gamma,-\delta}(\widehat{\xi}+(N-2)\gamma)-\chi_{\gamma,\varepsilon}(\undertilde{\xi}+(N-2)\gamma)}
{\chi_{\gamma,-\delta}(\widehat{\xi}+(N-2)\gamma)}\widehat{\undertilde{\boldsymbol{\psi}}}(N-1)\\
=&~\undertilde{\boldsymbol{\psi}}(N-2)-\frac{\chi_{\gamma,\varepsilon}(\undertilde{\xi}+(N-2)\gamma)}{\chi_{\gamma,-\delta}(\widehat{\xi}+(N-2)\gamma)}
\widehat{\boldsymbol{\psi}}(N-2),
\end{split}
\end{equation*}
we can get
\begin{equation*}
\begin{split}
&| \widehat{N-3}, ~ \widehat{\boldsymbol{\psi}}(N-2), ~\boldsymbol{\psi}(N-1)| |\widehat{N-2},  ~ \undertilde{\boldsymbol{\psi}}(N-2)|\\
-&| \widehat{N-3}, ~ \undertilde{\boldsymbol{\psi}}(N-2), ~ \boldsymbol{\psi}(N-1)| |\widehat{N-2}, ~  \widehat{\boldsymbol{\psi}}(N-2)|\\
-&|\widehat{N-1}|_{\boldsymbol{\psi}}| \widehat{N-3}, ~ \widehat{\boldsymbol{\psi}}(N-2), ~ \undertilde{\boldsymbol{\psi}}(N-2)|,
\end{split}
\end{equation*}
which is zero in light of Theorem \ref{P-C1}.
\end{proof}

We need to convert $\tau$ function from Casoratian to Hirota's form.
To achieve that, we investigate  $|\widehat{N-1}|$ generated by the basic column vector $\boldsymbol{\varphi}^-=(\varphi_1^-,\varphi_2^-,$
$\cdots, \varphi_N^-)^T$ with $\varphi_j^-$ defined in \eqref{phi-h1}.
\begin{lemma}\label{H1-L-2}
For the  basic column vector $\boldsymbol{\varphi}^-$, we have
\begin{align}
 |\widehat{N-1}|_{\boldsymbol{\varphi}^-}= &|\boldsymbol{\varphi}^-,~ E^3\boldsymbol{\varphi}^- ,~(E^3)^2 \boldsymbol{\varphi}^-,~ \cdots  ,~(E^3)^{N-1}\boldsymbol{\varphi}^- |\nonumber \\
=\,&(-1)^{N}\frac{\sigma(\xi-\sum_{i=1}^Nk_i)}{\sigma(\xi)}\cdot
\frac{\prod_{1\leq i<j\leq N}\sigma(k_i-k_j)}{\sigma^{N}(k_1)\cdots\sigma^{N}(k_N)}\prod_{i=1}^N\rho_{n,m,h}^+(k_i),\\
 |\widehat{N-2},N|_{\boldsymbol{\varphi}^-}=&|\boldsymbol{\varphi}^-,~ E^3\boldsymbol{\varphi}^- ,~(E^3)^2 \boldsymbol{\varphi}^-,~ \cdots  ,~(E^3)^{N-2}\boldsymbol{\varphi}^-,~(E^3)^{N}\boldsymbol{\varphi}^- |\nonumber \\
=&\left(\zeta(\xi-\sum_{i=1}^Nk_i)+\sum_{i=1}^N\zeta(k_i)+N\zeta(\gamma)-\zeta(\xi+N\gamma)\right)|\widehat{N-1}|_{\boldsymbol{\varphi^-}},\label{g-phi-}
\end{align}
where $\rho_{n,m,h}^+(k_i)$ is defined in \eqref{rho+-}.
\end{lemma}
\begin{proof}
For convenience, we first prove the following relation
\begin{equation*}
\begin{split}
&\left|\Phi_{\xi}(-\mathbf{k}) ,~ \Phi_{\xi+\gamma}(-\mathbf{k})\Phi_{\gamma}(\mathbf{k}),~
\cdots ,\Phi_{\xi+(N-1)\gamma}(-\mathbf{k})\Big(\Phi_{\gamma}(\mathbf{k})\Big)^{N-1}  \right| \\
=\,&(-1)^N\frac{\sigma(\xi-\sum_{i=1}^Nk_i)}{\sigma(\xi)}\cdot
\frac{\prod_{1\leq i<j\leq N}\sigma(k_i-k_j)}{\sigma^{N}(k_1)\cdots\sigma^{N}(k_N)},
\end{split}
\end{equation*}
where $\phi(\mathbf{k})$ denotes the column vector $(\phi(k_1), \cdots, \phi(k_N))^T$.
For the basic entries  $\Phi_{\xi+j\gamma}(-k_i)$ $\Big(\Phi_{\gamma}(k_i)\Big)^{j}$ in column, based on the definition of $\chi_{u,v}(z)$ \eqref{eq:add-4}
and the addition formula of Weierstrass $\sigma$ function \eqref{eq:add-5},
we can get
\begin{equation}\label{eq:add-7}
\begin{split}
\Phi_{\xi+j\gamma}(-k_i)\Big(\Phi_{\gamma}(k_i)\Big)^{j}&=\chi_{\gamma,(j-1)\gamma}(\xi)
\Phi_{\xi+(j-1)\gamma}(-k_i)\Big(\Phi_{\gamma}(k_i)\Big)^{j-1}\\
&-
\Phi_{(j-1)\gamma}(-k_i)\Phi_{\xi}(-k_i)\Big(\Phi_{\gamma}(k_i)\Big)^{j-1}, \quad j\geq2,
\end{split}
\end{equation}
and
\begin{equation*}
\Phi_{\xi+\gamma}(-k_i)\Phi_{\gamma}(k_i)
=\left(-\eta_{\xi}(\gamma)+\eta_{\xi}(-k_i)\right)\Phi_{\xi}(-k_i),
\end{equation*}
where $\eta_u(z)$ is defined as in  \eqref{eq:add-2}.
With this relations, it turns out that
\begin{equation*}
\begin{split}
   &\left|\Phi_{\xi}(-\mathbf{k}) ,~ \Phi_{\xi+\gamma}(-\mathbf{k})\Phi_{\gamma}(\mathbf{k}),
   ~
\cdots ,~\Phi_{\xi+(N-1)\gamma}(-\mathbf{k})\Big(\Phi_{\gamma}(\mathbf{k})\Big)^{N-1}  \right| \\
=&\prod_{i=1}^N\Phi_{\xi}(-k_i)\left|\mathbf{1}, ~\eta_{\xi}(-\mathbf{k}),~ -\Phi_{\gamma}(\mathbf{k})\Phi_{\gamma}(-\mathbf{k}),~
\cdots ,~ -\Big(\Phi_{\gamma}(\mathbf{k})\Big)^{N-2}\Phi_{(N-2)\gamma}(-\mathbf{k})  \right|.
\end{split}
\end{equation*}
For the column $-\Big(\Phi_{\gamma}(\mathbf{k})\Big)^{j}
\Phi_{j\gamma}(-\mathbf{k})$,
by setting $\mu=-k_i$, $\nu=\gamma$ in Eq.\eqref{eq:add-6},
 $ -\Big(\Phi_{\gamma}(k_i)\Big)^{j}$ $\Phi_{j\gamma}(-k_i)$ is further written into the linear combination of $\{\wp^{(s)}(-k_i)\}$ with $s=j-1, j-2, $ $j-3, \cdots, 2,1,0$,
where the coefficient of $\wp^{(s)}(-k_i)$  is a function of $\gamma$, and can be calculated with the formula \eqref{eq:FS-2}. In particular, the coefficient of $\wp^{(j-1)}(-k_i)$ is $\frac{1}{j!}$.
We can replace all such columns  from left to right,
and eliminate the linearly dependent columns with the front columns.
As a result, we have
\begin{align*}
&
\left(\prod_{j=1}^N\Phi_{\xi}(-k_j)\right)
\frac{1}{1! 2! \cdots (N-2)!}\\
&\times \left|\mathbf{1} ,  ~\eta_{\xi}(-\mathbf{k}), ~ \wp(-\mathbf{k}), ~ \wp'(-\mathbf{k}),~
  \wp''(-\mathbf{k}),~  \wp'''(-\mathbf{k}),~ \cdots ,~
  \wp^{(N-3)}(-\mathbf{k}) \right|\\
  =&(-1)^N\frac{\sigma(\xi-\sum_{i=1}^Nk_i)}{\sigma(\xi)}
\frac{\prod_{1\leq i<j\leq N}\sigma(k_i-k_j)}{\sigma^{N}(k_1)\cdots\sigma^{N}(k_N)},
\end{align*}
where we  have employed the  relation \eqref{eq:ND16} and the elliptic van der Monde determinant formula \eqref{eq:FS-1}.
The expression for $|\widehat{N-1}|_{\boldsymbol{\varphi}^-}$ is then obtained.

Next, based on \eqref{eq:add-7} we separate $|\widehat{N-2},N|_{\boldsymbol{\varphi}^-} $ into two terms,
\begin{equation*}
\begin{split}
   &|\widehat{N-2},N|_{\boldsymbol{\varphi}^-} =\chi_{\gamma,(N-1)\gamma}(\xi)|\widehat{N-1}|_{\boldsymbol{\varphi}^-}+\prod_{i=1}^N\rho^+_{n,m,h}(k_i)\prod_{i=1}^N\Phi_{\xi}(-k_i)\\
\times&\left|\mathbf{1} , ~   \eta_{\xi}(-\mathbf{k}),  ~ \wp(-\mathbf{k}),  ~\cdots,   \frac{1}{(N-3)!}\wp^{(N-4)}(-\mathbf{k}),~
-\Big(\Phi_{\gamma}(\mathbf{k})\Big)^{N-1}
\Phi_{(N-1)\gamma}(-\mathbf{k})  \right|.
\end{split}
\end{equation*}
Substituting the expression
\begin{equation*}
\begin{split}
 & -\Big(\Phi_{\gamma}(\mathbf{k})\Big)^{N-1}\Phi_{(N-1)\gamma}(-\mathbf{k})\\
  =& \frac{1}{(N-1)!}\wp^{(N-2)}(\mathbf{-k})
  -\frac{1}{(N-2)!}\wp^{(N-3)}(\mathbf{-k})\Big(\zeta((N-1)\gamma)-(N-1)\zeta(\gamma)\Big)+\cdots
  \end{split}
\end{equation*}
into the above equation,
making use of
\begin{equation*}
  \begin{split}
  &\left(\prod_{j=1}^N\Phi_{\xi}(-k_j)\right)
\frac{1}{1! 2! \cdots (N-3)!(N-1)!}\\
&\times |\mathbf{1},  ~\eta_{\xi}(-\mathbf{k}), ~ \wp(-\mathbf{k}), ~ \wp'(-\mathbf{k}),~
  \wp''(-\mathbf{k}),~   \cdots ,~
  \wp^{(N-4)}(-\mathbf{k}),~ \wp^{(N-2)}(-\mathbf{k}) |\\
  =& \left(\zeta\Big(\xi-\!\sum_{i=1}^Nk_i\Big)-\zeta(\xi)+\!\!\sum_{i=1}^N\zeta(k_i)\right)(-1)^N\frac{\sigma(\xi-\sum_{i=1}^Nk_i)}{\sigma(\xi)}
\frac{\prod_{1\leq i<j\leq N}\sigma(k_i-k_j)}{\sigma^{N}(k_1)\cdots\sigma^{N}(k_N)},
  \end{split}
\end{equation*}
we reach to
\begin{align*}
 |\widehat{N-2},N|_{\boldsymbol{\varphi}^-}
=\left(\zeta(\xi-\sum_{i=1}^Nk_i)
+\sum_{i=1}^N\zeta(k_i)+N\zeta(\gamma)-\zeta(\xi+N\gamma)\right)|\widehat{N-1}|_{\boldsymbol{\varphi}^-}.
\end{align*}
Thus Lemma \ref{H1-L-2} is proved.
\end{proof}

To proceed, we introduce $S=\{1,2,\cdots N\}$, $J$ being a subset of $S$,  and by $f^{\sharp}$ we denote
 $f^{\sharp}(\xi)=f(\xi+\sum_{i=1}^N k_i)$.
Hence we are led by the above lemma to the following.

\begin{lemma}
The $\tau$ function in Hirota's form   is defined by
\begin{subequations}
\begin{equation}\label{H1-f-Hirota}
\begin{split}
\tau&=\frac{f^{\sharp}}{f_0^{\sharp}}=\sum_{J\subset S}\tau_{J}\\
&=\sum_{J\subset S}\frac{\sigma(\xi+2\sum_{i \in J}k_i)}{\sigma(\xi)\prod_{i\in J}\sigma(2k_i)}\left(\prod_{i,j\in J \atop i<j}A_{ij}\right)
  \prod_{i\in J}\rho_{n,m,h}(k_i),
  \end{split}
  \end{equation}
where  $A_{ij}$ is defined as in \eqref{A-ij} and
\begin{equation}\label{H1-f-Hirota-a}
  \rho_{n,m,h}(k_i)=\left(\frac{\sigma(k_i-\delta)}{\sigma(k_i+\delta)}\right)^n\left(\frac{\sigma(k_i-\varepsilon)}{\sigma(k_i+\varepsilon)}\right)^m
\left(\frac{\sigma(k_i-\gamma)}{\sigma(k_i+\gamma)}\right)^h\rho_{0,0,0}(k_i).
\end{equation}
\end{subequations}
\end{lemma}
\begin{proof}
Similar to the KdV case, $f=\sigma(\xi)|\widehat{N-1}|$ can be split and then written as a sum of $2^N$ terms in Casoratian form, the basic entries of each Casoratian are $\phi_j\!=\!(\nu_j)^j\Phi_{\xi}(\nu_jk_j)\rho_{n,m,h}^-(\nu_jk_j)$, in which $\nu_j$ run over $\{1,-1\}$. $\phi_j$ can be obtained from $\varphi_j^-$ by multiplying  $(\nu_j)^j$
and replacing $k_j$ with $-\nu_j k_j$,  $\boldsymbol{\nu}$ denotes cluster
$\boldsymbol{\nu}=\{\nu_1, \nu_2, \cdots, \nu_N\}$ and
$|\boldsymbol{\nu}|$  denotes the length of $\boldsymbol{\nu}$, that is  the number of
positive $\nu_j$'s in the cluster $\boldsymbol{\nu}$. Thus we have
\begin{equation}
f=\sum^{N}_{l=0} \sum_{|\boldsymbol{\nu}|=l}f_{\boldsymbol{\nu}}=\sum^{N}_{l=0}f_l,
\end{equation}
where $f_{\boldsymbol{\nu}}$, in light of Lemma \ref{H1-L-2},  can be expressed as
\begin{equation}
\begin{split}
f_{\boldsymbol{\nu}}
= &(-1)^N\prod^{N}_{j=1}(\nu_j)^j\cdot
\sigma(\xi+\sum_{i=1}^N \nu_i k_i)\cdot
\frac{\prod_{1\leq i<j\leq N}\sigma(-\nu_i k_i+\nu_j k_j)}
{\sigma^{N}(-\nu_1 k_1)\cdots\sigma^{N}(-\nu_N k_N)}\prod_{i=1}^N \rho^-_{n,m,h}(\nu_ik_i)\\
=& (-1)^\frac{N(N-1)}{2}\cdot
\sigma(\xi+\sum_{i=1}^N \nu_i k_i)\cdot
\frac{\prod_{1\leq i<j\leq N}\nu_i\,\sigma(\nu_i k_i-\nu_j k_j)}
{\sigma^{N}(k_1)\cdots\sigma^{N}(k_N)}\prod_{i=1}^N \rho^-_{n,m,h}(\nu_ik_i),\label{W-phi-eps}
\end{split}
\end{equation}
and in particular, $f_0$ is
\begin{align}\label{h1-f-}
 f_0=&(-1)^{\frac{N(N-1)}{2}}\sigma(\xi-\sum_{i=1}^Nk_i)\cdot
\frac{\prod_{1\leq i<j\leq N}\sigma(k_i-k_j)}{\sigma^{N}(k_1)\cdots\sigma^{N}(k_N)}\prod_{i=1}^N\rho^+_{n,m,h}(k_i),
\end{align}
where $\rho^{\pm}_{n,m,h}(k_i)$ are defined as \eqref{rho+-}.
Consider such $J$ where all $\nu_i=1$,  and consequently in  $S\backslash J$ all $\nu_i=-1$.
Then we have
\begin{equation*}
\begin{split}
  \frac{f_{\boldsymbol{\nu}}^{\sharp}}{f_0^{\sharp}}=   \tau_J &=\frac{\sigma(\xi+2\sum_{i\in J}k_i)}{\sigma(\xi)}\prod_{i\in J \atop j\in S\backslash J}\frac{\sigma(k_i+k_j)}{\sigma(k_i-k_j)\cdot \rm{sgn}[j-i]}\\
  &\times\prod_{i\in J}\left(\frac{\sigma(k_i-\delta)}{\sigma(k_i+\delta)}\right)^n\left(\frac{\sigma(k_i-\varepsilon)}{\sigma(k_i+\varepsilon)}\right)^m
\left(\frac{\sigma(k_i-\gamma)}{\sigma(k_i+\gamma)}\right)^h.
\end{split}
\end{equation*}
By defining
\begin{equation*}
  \rho_{0,0,0}(k_i)=\sigma(2k_i)\prod_{1\leq j\leq N \atop j\neq i }\frac{\sigma(k_i+k_j)}{\sigma(k_i-k_j)\rm{sgn[j-i]}},
\end{equation*}
and noting that
 \begin{equation*}
 \begin{split}
  &\prod_{i\in J \atop j\in S\backslash J}\frac{\sigma(k_i+k_j)}{\sigma(k_i-k_j)\cdot \rm{sgn}[j-i]} \cdot
   \prod_{i\in J}\prod_{1\leq j\leq N \atop j\neq i }\frac{\sigma(k_i-k_j)\rm{sgn[j-i]}}{\sigma(k_i+k_j)}\\
  =&\prod_{i,j\in J\atop  i<j}
\left(\frac{\sigma(k_{i}- k_j)}{\sigma(k_{i}+ k_j)}\right)^2
=\prod_{i,j\in J\atop  i<j} A_{ij},
  \end{split}
 \end{equation*}
we have
 \begin{equation*}
 \tau_J=\frac{\sigma(\xi+2\sum_{i \in J}k_i)}{\sigma(\xi)\prod_{i\in J}\sigma(2k_i)}\left(\prod_{i,j\in J \atop i<j}A_{ij}\right)
  \prod_{i\in J}\rho_{n,m,h}(k_i),
 \end{equation*}
which gives rise to the $\tau$ function defined by \eqref{H1-f-Hirota} coupled with \eqref{H1-f-Hirota-a}.

\end{proof}

Next, we introduce a function
  \begin{equation}\label{H1-p-Hirota}
  \begin{split}
p&=\frac{\left(\zeta(\xi)+N\zeta(\gamma)-\zeta(\xi+N\gamma+\sum_{i=1}^N k_i)\right)f^{\sharp}-g^{\sharp}}{f_0^{\sharp}}.
\end{split}
\end{equation}
Performing a similar computation as we have done for $\tau$ in the above lemma,  we can express
$p$ in terms of $\tau_J$:
\begin{equation*}
\begin{split}
p&=\left(\zeta(\xi)+N\zeta(\gamma)-\zeta(\xi+N\gamma+\sum_{i=1}^N k_i)\right)\frac{f^{\sharp}}{f_0^{\sharp}}\\
&-\sum_{J\subset S}\left(\zeta(\xi+2\sum_{i\in J}k_i)- \! \sum_{i\in J}\zeta(k_i)+  \!\!\! \sum_{j\in  S\backslash J}\zeta(k_j)+N\zeta(\gamma)-\zeta(\xi+N\gamma+ \! \!\sum_{i=1}^N k_i)\right)\tau_{J}\\
&=\sum_{J\subset S}\left(\zeta(\xi)-\zeta(\xi+2\sum_{i \in J}k_i)+\sum_{i\in J}\zeta(k_i)-\sum_{j\in  S\backslash J}\zeta(k_j)\right)\tau_{J}.
\end{split}
\end{equation*}
Then we have the following bilinear system w.r.t. $\tau$ and $p$.

\begin{theorem}\label{th-h1}
The bilinear lpKdV system \eqref{4.11}  can be represented in terms of $\tau$ and $p$:
\begin{subequations}\label{eq:bilinear-h1}
\begin{align}
\chi_{-\varepsilon,\delta}(\widehat{\xi})\left(\widetilde{\tau}\widehat{\tau}-\tau\widehat{\widetilde{\tau}}\right)+\widetilde{p}\widehat{\tau}-\widehat{p}\widetilde{\tau}
&=0,\\
\chi_{\varepsilon,\delta}(\xi)\left(\tau\widehat{\widetilde{\tau}}-\widetilde{\tau}\widehat{\tau}\right)+\widehat{\widetilde{p}}\tau
-p\widehat{\widetilde{\tau}}&=0,
\end{align}
\end{subequations}
where $\tau$ and $p$ are defined in \eqref{H1-f-Hirota} and \eqref{H1-p-Hirota}.
\end{theorem}
\begin{proof}
After the transformation $\xi\to\xi+\sum_{i=1}^N k_i$,
the elliptic $N$-soliton solution \eqref{seed-H1} is rewritten  as
\begin{equation*}
\begin{split}
w^{\sharp}
&=w_0-v,
\end{split}
\end{equation*}
where $w_0$ is given in \eqref{slt:w0}
and
\begin{equation*}
\begin{split}
v&=\zeta(\xi)+N\zeta(\gamma)-\zeta(\xi+N\gamma+\sum_{i=1}^N k_i)-\frac{g^{\sharp}}{f^{\sharp}}.
\end{split}
\end{equation*}
Note that  $f^{\sharp}=\tau f_0^{\sharp}$ and $\left(\zeta(\xi)+N\zeta(\gamma)-\zeta(\xi+N\gamma+\sum_{i=1}^N k_i)\right)f^{\sharp}-g^{\sharp}=p f_0^{\sharp}$, the bilinear equations \eqref{4.11} can be rewritten as
\begin{subequations}\label{xx}
\begin{align}
\widetilde{f_0^{\sharp}}\widehat{f_0^{\sharp}}\left(\chi_{-\varepsilon,\delta}(\widehat{\xi})\widetilde{\tau}\widehat{\tau}+\widetilde{p}\widehat{\tau}
-\widehat{p}\widetilde{\tau}\right)
&=f_0^{\sharp}\widehat{\widetilde{f_0^{\sharp}}}\Phi_{\delta}(-\varepsilon)\tau\widehat{\widetilde{\tau}},\\
f_0^{\sharp}\widehat{\widetilde{f_0^{\sharp}}}\left(\chi_{\varepsilon,\delta}(\xi)\tau\widehat{\widetilde{\tau}}+\widehat{\widetilde{p}}\tau
-p\widehat{\widetilde{\tau}}\right)
&=\widetilde{f_0^{\sharp}}\widehat{f_0^{\sharp}}\Phi_{\delta}(\varepsilon)\widetilde{\tau}\widehat{\tau}.
\end{align}
\end{subequations}
Recalling the expression \eqref{h1-f-} of $f_0$ and making use of the following relations
\begin{equation}\label{eq:add-8}
  \Phi_{\delta}(-\varepsilon)\frac{\sigma(\xi)\sigma(\xi+\delta+\varepsilon)}{\sigma(\xi+\delta)\sigma(\xi+\varepsilon)}
  =\chi_{-\varepsilon,\delta}(\xi+\varepsilon), ~~
  \Phi_{\delta}(\varepsilon)\frac{\sigma(\xi+\delta)\sigma(\xi+\varepsilon)}{\sigma(\xi)\sigma(\xi+\delta+\varepsilon)}
  =\chi_{\varepsilon,\delta}(\xi),
\end{equation}
\eqref{xx} yields the bilinear equations \eqref{eq:bilinear-h1}.

\end{proof}

\begin{remark}
The $\tau$ function \eqref{H1-f-Hirota} can be generated by the vertex operator \eqref{eq:vertexkdv}.
In fact, redefining $c_i=\left(\frac{\sigma(k_i-\gamma)}{\sigma(k_i+\gamma)}\right)^h\rho_{0,0,0}(k_i)$
and introducing Miwa's coordinates \cite{Miwa82}
\begin{equation}
t_{2j+1}=\frac{\delta^{2j+1}n+\varepsilon^{2j+1}m}{2j+1},
\end{equation}
the $\tau$ function \eqref{H1-f-Hirota} exactly has an expression \eqref{tau-ver},
which can be generated by the vertex operator \eqref{eq:vertexkdv} as in Theorem \ref{T-3}.
\end{remark}

Finally, let us write down the explicit form of the elliptic 1-soliton solution. After a shift $\xi\to\xi+k_1$
\begin{equation*}
\begin{split}
  w^{\sharp}&=w_0-v\\
  &=w_0-\left(\chi_{\gamma,k_1}(\xi)-\frac{g^{\sharp}}{f^{\sharp}}\right),
  \end{split}
\end{equation*}
where $w_0=\zeta(\xi)-n\zeta(\delta)-m\zeta(\varepsilon)-h\zeta(\gamma)+\zeta(k_1)-\zeta(\xi_0)$.
Here we deal with the $p$ function by adding $\zeta(k_1)\tau$  to $p$  in \eqref{H1-p-Hirota}, which does not change the bilinear equations \eqref{eq:bilinear-h1}.
Thus, by using \eqref{H1-f-Hirota}, we can get
\begin{align*}
\tau&=1+\rho_{n,m,h}(k_1)\Phi_{\xi}(2k_1),\\
p&=\Phi_{\xi}(k_1)^2\rho_{n,m,h}(k_1),
\end{align*}
in which we redefine $  \rho_{n,m,h}(k_1)$ by absorbing $\rho_{0,0,0}(k_1)=\sigma(2k_1)$, i.e.
\begin{equation*}
  \rho_{n,m,h}(k_1)=\left(\frac{\sigma(k_1-\delta)}{\sigma(k_1+\delta)}\right)^n\left(\frac{\sigma(k_1-\varepsilon)}{\sigma(k_1+\varepsilon)}\right)^m
\left(\frac{\sigma(k_1-\gamma)}{\sigma(k_1+\gamma)}\right)^h\rho_{0,0,0}(k_1).
\end{equation*}
The expression of $v=p/\tau$  is formally  same as the elliptic 1-soliton solution
obtained from the Cauchy matrix approach (see  Eq.(4.33) in \cite{IMRN2010}).

\section{Elliptic \texorpdfstring{$\tau$}{} function of lattice  potential KP equation}\label{lkp}

In this section, we investigate the lattice  potential KP  (lpKP) equation
\begin{equation}\label{eq:lkp}
(\widehat{\overline{w}}-\widehat{\widetilde{w}})(\overline{w}-\widehat{w})=(\widetilde{\overline{w}}-\widehat{\widetilde{w}})(\overline{w}-\widetilde{w}),
\end{equation}
which is one of the  most important $3$-dimensional lattice equations \cite{NCWQ84}.
In \cite{YN-JMP-2013}, the elliptic  soliton solution of lpKP equation was derived
using the Cauchy matrix method, and later, an elliptic  scheme of direct linearisation was established in \cite{22NSZ}.
As discussed  in the previous sections, we will give the bilinear form and $\tau$ function of the elliptic solitons
of the lpKP equation. The connection between the $\tau$ function and  vertex operator will be addressed as well.
\begin{lemma}
	The elliptic N-soliton solution in Casoratian form  of the lpKP equation is given by
\begin{equation}\label{5.2}
\begin{split}
w&=\zeta(\xi+N\gamma)-N\zeta(\gamma)
-n\zeta(\delta)-m\zeta(\varepsilon)-h\zeta(\gamma)-\zeta(\xi_0)+\frac{g}{f},
\end{split}
\end{equation}
 where
\begin{equation}\label{slt:LKPCas}
\begin{split}
f=\sigma(\xi)|\widehat{N-1}|, \qquad g=\sigma(\xi)|\widehat{N-2},N|,
\end{split}
\end{equation}
and the generating  column vector $\boldsymbol{\phi}=(\phi_1,\cdots,\phi_N)^T$ is composed by
\begin{equation}\label{5.4}
\begin{split}
\phi_{i}=\rho_{n,m,h}^{-}(k_i)\Phi_{\xi}(k_i)+\rho_{n,m,h}^-(l_i)\Phi_{\xi}(l_i),
\end{split}
\end{equation}
with $\rho_{n,m,l}^{-}(z)$ is defined in \eqref{rho+-}.
\end{lemma}

\begin{proof}
Introduce the bilinear equations
\begin{subequations}\label{bilinear-LKP}
\begin{align}
\mathcal{H}_1 &\equiv \chi_{\varepsilon,-\gamma}(\xi+(N+1)\gamma)\overline{f}\widehat{f}+\overline{g}\widehat{f}-\widehat{g}\overline{f}
-\Phi_{\varepsilon}(-\gamma)f\widehat{\overline{f}}=0,\\
\mathcal{H}_2 &\equiv\chi_{\delta,-\gamma}(\xi+(N+1)\gamma)\widetilde{f}\overline{f}+\overline{g}\widetilde{f}-\widetilde{g}\overline{f}
-\Phi_{\delta}(-\gamma)f\widetilde{\overline{f}}=0.
\end{align}
\end{subequations}
Using \eqref{5.2},
 the lpKP equation can be rewritten as
\begin{equation}
\left(\frac{\mathcal{H}_1+\Phi_{\varepsilon}(-\gamma)f\widehat{\overline{f}}}{\overline{f}\widehat{f}}
\right)^{\widetilde{\,}} \!\!\!
\left(\frac{\mathcal{H}_2+\Phi_{\delta}(-\gamma)f\widetilde{\overline{f}}}{\overline{f}\widetilde{f}}\right)
\!\!=\!\!\left(\frac{\mathcal{H}_1+\Phi_{\varepsilon}(-\gamma)f\widehat{\overline{f}}}{\overline{f}\widehat{f}}\right) \!\!\!
\left(\frac{\mathcal{H}_2+\Phi_{\delta}(-\gamma)f\widetilde{\overline{f}}}
{\overline{f}\widetilde{f}}\right)^{\widehat{\,}}.
\end{equation}
Thus we only need to prove that $f$ and $g$ satisfy $\mathcal{H}_1=0$, and $\mathcal{H}_2=0$.
The functions  $\phi_i$ defined in \eqref{5.4}
obey the same shift relations as in Eqs.\eqref{eq:varphi-shift-1}
which means the shift relations of $f, g$ in tilde and hat directions are the same as those of the lpKdV case.
Below we list out necessary shift relations in bar direction
\begin{align*}
\underline{f}&=\sigma(\underline{\xi})|\underline{\boldsymbol{\phi}}, ~ \widehat{N-2}|,  &\,
&\underline{g}=\sigma(\underline{\xi})|\underline{\boldsymbol{\phi}},~ \widehat{N-3}, ~ \boldsymbol{\phi}(N-1)|,\\
\undertilde{\underline{f}}&=\sigma(\undertilde{\underline{\xi}})\frac{\Big|\underline{\boldsymbol{\phi}},~ \widehat{N-3},~ \undertilde{\boldsymbol{\phi}}(N-2)\Big|}
{(-1)^{N-1}\prod_{j=-1}^{N-3}\chi_{\gamma,-\delta}(\xi+j\gamma)}, &\,
&\underset{\widehat{}}{\underline{f}}=\sigma(\underset{\widehat{}}{\underline{\xi}})
\frac{\Big|\underline{\boldsymbol{\phi}},~ \widehat{N-3},~ \underset{\widehat{}}{\boldsymbol{\phi}}(N-2)\Big|}
{(-1)^{N-1}\prod_{j=-1}^{N-3}\chi_{\gamma,-\varepsilon}(\xi+j\gamma)}.
\end{align*}
Similar to the proof of Theorem \ref{th-h1}, $\mathcal{H}_1=0$ and $\mathcal{H}_2=0$ can be proved  in down-hat-bar-shifted and down-tilde-bar-shifted version, respectively.
We skip the details.
\end{proof}

\begin{sloppypar}
\begin{remark}
Since the lpKP equation \eqref{eq:lkp}  is 4-dimensionally consistent \cite{ABS-2012},
apart from the bilinear equations $\mathcal{H}_1$ and $\mathcal{H}_2$, the following holds,
\begin{equation}
\mathcal{H}_3\equiv \chi_{\delta,-\varepsilon}(\widehat{\xi}+N\gamma)\widetilde{f}\widehat{f}+\widehat{g}\widetilde{f}-\widetilde{g}\widehat{f}
-\Phi_{\delta}(-\varepsilon)f\widetilde{\widehat{f}}.
\end{equation}
It is easy to verify that $\mathcal{H}_1/\left(\overline{f}\widehat{f}\right)-\mathcal{H}_2\left(\overline{f}\widetilde{f}\right)+\mathcal{H}_3/\left(\widehat{f}\widetilde{f}\right) =0$
and this yields the Hirota-Miwa equation
\begin{equation}\label{eq:HM}
\Phi_{\varepsilon}(-\gamma)\widetilde{f} ~\widehat{\overline{f}}+\Phi_{\gamma}(-\delta)\widehat{f} ~\widetilde{\overline{f}}
+\Phi_{\delta}(-\varepsilon)\overline{f}\widehat{\widetilde{f}}=0.
\end{equation}
From the identity \eqref{eq:add-5} one can immediately find that $\sigma(\xi)$ is a solution of the above equation.
\end{remark}
\end{sloppypar}

Let $J$ be a subset of $S=\{1,2,\cdots,N\}$, $\psi_j=\Phi_{\xi}(k_j)\rho_{n,m,h}^-(k_j)$  when $j\in J$
and $\psi_j=\Phi_{\xi}(l_j) \rho_{n,m,h}^-(l_j)$ when $j\in S\backslash J$, by $f_J$ we denote  $\sigma(\xi)|\widehat{N-1}|$ generated by
\begin{equation}\label{vphi-J-LKP}
\boldsymbol{\psi}=(\psi_1, \psi_2, \cdots, \psi_N)^T.
\end{equation}
 According to Lemma \ref{H1-L-2}, one can get
\begin{align}
f_{J}^{}=&(-1)^{\frac{N(N-1)}{2}}
\sigma(\xi+\sum_{i\in J }k_{i}+\sum_{j\in S\backslash J}l_j)
\frac{\prod_{i\in J \atop j\in S\backslash J}\sigma(k_i-l_j)\mathrm{sgn}[j-i]}
{\left(\prod_{i\in J}\sigma^N(k_{i})\right)\left(\prod_{ j\in S\backslash J}\sigma^N(l_{j})\right)} \nonumber\\
 & \times \left(\prod_{i<j\in J}\sigma(k_i-k_j) \right)
 \left(\prod_{i<j\in S\backslash J}\sigma(l_{i}-l_j)\right)
\prod_{i\in J}\rho^-_{n,m,h}(k_i)\prod_{j\in S\backslash J}\rho^-_{n,m,h}(l_j),
\end{align}
and
\begin{equation}\label{g-LKP}
f_{\varnothing}
=(-1)^{\frac{N(N-1)}{2}}
\sigma(\xi+\sum_{j\in S}l_j)
\frac{\prod_{i<j \in S}\sigma(l_i-l_j)}
{\prod_{ j\in S}\sigma^N(l_{j})} \prod_{j=1}^N\rho^-_{n,m,h}(l_j).
\end{equation}

After a shift  $f^{\natural}(\xi)=f(\xi-\sum_{j=1}^N l_j)$, we can get the $\tau$ function in Hirota's form.

\begin{theorem}
For the function $f$ in Casoratian form \eqref{slt:LKPCas} and $f_{\varnothing}$ ,
\begin{equation}\label{taup-LKP}
\tau=\frac{f^{\natural}}{f_{\varnothing}^{\natural}}, \qquad
 p= \frac{\left(\zeta(\xi)+N\zeta(\gamma)-\zeta(\xi+N\gamma-\sum_{j=1}^N l_j)\right)f^{\natural}-g^{\natural}}{f_{\varnothing}^{\natural}}
\end{equation}
solve the bilinear equations
\begin{subequations}\label{bilinear-LKP-2}
\begin{align}
\mathcal{H}'_1 \equiv &\chi_{\varepsilon,-\gamma}(\overline{\xi})\left(\overline{\tau}\widehat{\tau}-\tau\widehat{\overline{\tau}}\right)
+\widehat{p}\overline{\tau}-\overline{p}\widehat{\tau}
=0,\label{bilinear-LKP-a}\\
\mathcal{H}'_2 \equiv&\chi_{\delta,-\gamma}(\overline{\xi})\left(\widetilde{\tau}\overline{\tau}-\tau\widetilde{\overline{\tau}}\right)
+\widetilde{p}\overline{\tau}-\overline{p}\widetilde{\tau}
=0,\label{bilinear-LKP-b}
\end{align}
\end{subequations}
and $\tau$ is written in Hirota's form as
\begin{equation}\label{Hirota-LKP}
\begin{split}
\tau=&\sum_{J \subset S}\tau_J\\
=&\sum_{J \subset S} \frac{\sigma(\xi+\sum_{i\in J}(k_i-l_i))}{\sigma(\xi)\prod_{i\in J}\sigma(k_i-l_i)}
\left(\prod_{i<j\in J}\frac{\sigma(k_i-k_j)\sigma(l_i- l_j)}{\sigma(k_i- l_j)\sigma(l_i- k_j)}\right) \prod_{i\in J}\rho_{n,m,h}(k_i),
\end{split}
\end{equation}
and
\begin{equation}\label{p-LKP}
\begin{split}
p=&\sum_{J\subset S}\left(\zeta(\xi)-\zeta\Big(\xi+\sum_{i\in J}(k_i-l_i)\Big)+\sum_{i\in J}\zeta(k_i)+\sum_{j\in S\backslash J}\zeta(l_j)\right)\tau_J,
\end{split}
\end{equation}
where
\begin{equation*}
\rho_{n,m,h}(k_i)=\left(\frac{\sigma(k_i-\delta)\sigma(l_i)}{\sigma(k_i)\sigma(l_i-\delta)}\right)^n\left(\frac{\sigma(k_i-\varepsilon)\sigma(l_i)}{\sigma(k_i)\sigma(l_i-\varepsilon)}\right)^m
\left(\frac{\sigma(k_i-\gamma)\sigma(l_i)}{\sigma(k_i)\sigma(l_i-\gamma)}\right)^h\rho_{0,0,0}(k_i).
\end{equation*}
\end{theorem}
\begin{proof}
The proof  is similar to the  lpKdV and KP cases \cite{22LxZhang}, so we just briefly sketch the proof.
We implement a shift $\xi \to \xi-\sum_{i=1}^{N}l_i$ in Eqs.\eqref{bilinear-LKP}
and express it in terms of $\tau$ and $p$ defined in \eqref{taup-LKP}.
Then,
making use of identity \eqref{eq:add-8}, we can obtain  the bilinear equations \eqref{bilinear-LKP-2}.

Next, the general term $\tau_J$ in $\tau$ is
  \begin{equation*}
\begin{split}
\frac{f_{J}^{\natural}}{f_{\varnothing}^{\natural}}=& \frac{\sigma(\xi+\sum_{i\in J}(k_i-l_i))}{\sigma(\xi)}
\left(\prod_{i<j\in J}\frac{\sigma(k_i-k_j)}{\sigma(l_i- l_j)}\right)\left(\prod_{i\in J}\frac{\sigma^N(l_i)}{\sigma^N(k_i)}
\prod_{j\in S\setminus J}\frac{\sigma(k_{i}- l_j)}{\sigma(l_{i}-l_j)}\right)\\
\times &\prod_{i\in J}\left(\frac{\sigma(k_i-\delta)\sigma(l_i)}{\sigma(k_i)\sigma(l_i-\delta)}\right)^n\left(\frac{\sigma(k_i-\varepsilon)\sigma(l_i)}{\sigma(k_i)\sigma(l_i-\varepsilon)}\right)^m
\left(\frac{\sigma(k_i-\gamma)\sigma(l_i)}{\sigma(k_i)\sigma(l_i-\gamma)}\right)^h.
\end{split}
\end{equation*}
After redefining
\begin{equation*}
\rho_{0,0,0}(k_i)=\sigma(k_i-l_i)\frac{\sigma^N(l_{i})}{\sigma^N(k_{i})}
\prod_{j\in S\atop j\neq i}\frac{\sigma(k_{i}- l_j)}{\sigma(l_{i}-l_j)},
\end{equation*}
we find that $\tau$ takes the explicit Hirota's form \eqref{Hirota-LKP}.
$p$ can be obtained in a similar way.

\end{proof}

\begin{remark}
The $\tau$ functions \eqref{Hirota-LKP} and \eqref{tau-KP} share the same vertex operator
\eqref{eq:vertexkp} in light of Miwa's coordinates
\begin{equation}
t_{j}=\frac{\delta^{j}n+\varepsilon^{j}m+\gamma^{j}h}{j}.
\end{equation}
In addition, due to the 4-dimensional consistency of the lpKP equation, apart from the bilinear equations
\eqref{bilinear-LKP-2}, we also have
\begin{equation}
  \mathcal{H}'_3\equiv \chi_{\delta,-\varepsilon}(\widehat{\xi})\left(\widetilde{\tau}\widehat{\tau}-\tau\widetilde{\widehat{\tau}}\right)
  +\widetilde{p}\widehat{\tau}-\widehat{p}\widetilde{\tau}=0.
\end{equation}
It follows from the identity $\mathcal{H}'_1/\left(\overline{\tau}\widehat{\tau}\right)-\mathcal{H}'_2\left(\overline{\tau}\widetilde{\tau}\right)+\mathcal{H}'_3/\left(\widehat{\tau}\widetilde{\tau}\right) =0$
that
\begin{equation}
\chi_{\varepsilon,-\gamma}(\overline{\xi})\widetilde{\tau}\widehat{\overline{\tau}}
+\chi_{\gamma,-\delta}(\widetilde{\xi})\widehat{\tau}\widetilde{\overline{\tau}}
+\chi_{\delta,-\varepsilon}(\widehat{\xi})\overline{\tau}\widehat{\widetilde{\tau}}=0,
\end{equation}
which indicates that $\sigma(\xi)\tau$ is the solution to the Hirota-Miwa equation \eqref{eq:HM}.
\end{remark}

\section{Conclusions}\label{sec-6}

In this chapter we have  reviewed recent development in the Hirota bilinear method on elliptic solitons
of the KdV equation and KP equation, including bilinear calculations involved with  the Lam\'e type PWFs,
expressions of $\tau$ functions and the generating vertex operators.
More details can be found in \cite{22LxZhang}.
Elliptic solitons composed of discrete analogue of the Lam\'e type PWFs
can be constructed using Cauchy matrix approach \cite{IMRN2010} and
an elliptic direct linearisation scheme \cite{22NSZ}.
In this chapter, we developed the Hirota bilinear method on the elliptic solitons for
discrete integrable systems.
For the lpKdV equation and lpKP equation,
we presented their bilinear forms, derived $\tau$ functions of elliptic solitons in Casoratian form
and Hirota's form.
In light of Miwa's correspondence between discrete and continuous independent variables,
we can show that the elliptic $\tau$ functions of the lpKdV equation
and the lpKP equation share the same vertex operators with the continuous
KdV hierarchy and the KP hierarchy, respectively.
This fact will lead us to the further research of the elliptic $\tau$ functions along the line of Date-Jimbo-Miwa's work \cite{DJM-JPSJ-82,DJM-JPSJ-83}.

\vskip 30pt

\subsection*{Acknowledgments}
X. Li was supported  by the Foundation of Jiangsu Normal University grant no. 22XFRS031.
D.J. Zhang was supported by the NSFC grant (nos. 12271334, 12126352).


\begin{thebibliography}{99}
\bibitem[ABS03]{ABS-2003} V.E. Adler, A.I. Bobenko, Y.B. Suris,
         \textit{Classification of integrable equations on quad-graphs. The consistency approach},
        Commun. Math. Phys., \textbf{233} (2003),  513--543.

\bibitem[ABS12]{ABS-2012} V.E. Adler, A.I. Bobenko, Y.B. Suris,
         \textit{Classification of integrable discrete equations of octahedron type},
        Int. Math. Res. Notices, \textbf{2012} (2012),  1882--1889.

\bibitem[DJM82]{DJM-JPSJ-82} 	E. Date, M. Jimbo, T. Miwa,
           \textit{Method for generating discrete soliton equations, I},
           J. Phys. Soc. Japan, \textbf{51} (1982), 4116--4124;
          \textit{II},
          J. Phys. Soc. Japan, \textbf{51} (1982), 4125--4131.

\bibitem[DJM83]{DJM-JPSJ-83}  E. Date, M. Jimbo, T. Miwa,
          \textit{Method for generating discrete soliton equations, III},
          J. Phys. Soc. Japan, \textbf{52} (1983), 388--393;
          \textit{IV},
          J. Phys. Soc. Japan, \textbf{52} (1983), 761--765;
          \textit{V},
          J. Phys. Soc. Japan, \textbf{52} (1983), 766--771.


\bibitem[FN83]{FreN-1983} N.C. Freeman, J.J.C. Nimmo,
          \textit{Soliton solutions of the KdV and KP equations: the Wronskian technique},
         Phys. Lett., \textbf{95A} (1983), 1--3.


\bibitem[FS80]{FS80} G. Frobenius, L. Stickelberger,
         \textit{Ueber die addition und multiplication der elliptischen functionen},
         J. Reine Angew. Math., \textbf{88} (1880), 146--184.



\bibitem[HJN16]{HJN-book}
J. Hietarinta, N. Joshi, F.W. Nijhoff,
           \textit{ Discrete Systems and Integrability,}
             Camb. Univ. Press, Cambridge, 2016.

\bibitem[HZ09]{09HZ}J. Hietarinta, D.J. Zhang,
           \textit{ Soliton solutions for ABS lattice equations: II Casoratians and bilinearization},
            J. Phys. A: Math. Theor., \textbf{42} (2009), No.404006 (30pp).

\bibitem[Hir74]{74Hirota} R. Hirota,
            \textit{A new form of B\"{a}cklund transformations and its relation to the inverse scattering problem},
           Prog. Theor. Phys., \textbf{52} (1974) 1498--1512.


\bibitem[Ince40]{Ince-1940} E.L. Ince,
         \textit{Further investigations into the periodic Lam\'e functions},
         Proc. Roy. Soc. Edinburgh, \textbf{60} (1940), 83--99.

\bibitem[Kie73]{Kiepert-1873}
L. Kiepert, \textit{Wirkliche Ausf\"uhrung der ganzzahligen Multiplikation der elliptichen Funktionen}, J. reine
angew. Math. \textbf{76} (1873), 21--33.


\bibitem[KM74]{74KM}
E. A. Kuznetsov, A.V. $\rm{Mikha\breve{i}lov}$,
 \textit{Stability of stationary waves in nonlinear weakly dispersive media},
Sov. Phys. JETP, \textbf{40} (1974), 855--859.
(Zh. Eksp. Teor. Fiz., \textbf{67} (1974), 1717--1727.)


\bibitem[LZ22]{22LxZhang}
X. Li, D.J. Zhang, \textit{Elliptic soliton solutions: $\tau$ functions, vertex operators and bilinear
identities}, J. Nonlinear Sci., \textbf{32} (2022), No.70 (53pp).

\bibitem[Mat08]{Matveev-08} V.B. Matveev,
          \textit{30 years of finite-gap integration theory},
          Philos. Trans. R. Soc. A, \textbf{366} (2008), 837--875.


\bibitem[Miwa82]{Miwa82}
T. Miwa, \textit{On Hirota’s difference equations}, Proc. Japan Acad. A, \textbf{58} (1982), 9--12.


\bibitem[MJD99]{MJD-book-1999}  T. Miwa, M. Jimbo, E. Date,
          \textit{Solitons: Differential Equations, Symmetries and Infinite Dimensional Algebras},
          Camb. Univ. Press, Cambridge, 1999.

\bibitem[NA10]{IMRN2010}F.W. Nijhoff, J. Atkinson,
           \textit{Elliptic $N$-soliton solutions of ABS lattice equations},
          Int. Math. Res. Not., \textbf{2010} (2010), 3837--3895.



\bibitem[NCWQ84]{NCWQ84}
F.W. Nijhoff, H.W. Capel, G.L. Wiersma, G.R.W. Quispel,  \textit{B\"acklund transformations and three-dimensional
lattice equations,} Phys. Lett. A, \textbf{105} (1984), 267--272.


\bibitem[ND18]{ND-2016} F.W. Nijhoff, N. Delice,
        \textit{On elliptic Lax pairs and isomonodromic deformation systems for elliptic lattice equations},
       Adv. Stud. Pure Math, \textbf{76} (2018), 487--525.

\bibitem[NQC83]{83-NQC}
 F.W. Nijhoff, G.R.W. Quispel, H.W. Capel, \textit{Direct linearization of nonlinear
difference-difference equations},  Phys. Lett. A, \textbf{97} (1983), 125--128.

\bibitem[NSZ23]{22NSZ}
F.W. Nĳhoff, Y.Y. Sun, D.J. Zhang,
 \textit{Elliptic solutions of Boussinesq type lattice equations
and the elliptic N-th root of unity},
 Commun. Math. Phys., \textbf{399} 
  (2023), 599--650.



\bibitem[YN13]{YN-JMP-2013} S. Yoo-Kong, F.W. Nijhoff,
         \textit{ Elliptic $(N, N')$-soliton solutions of the lattice Kadomtsev-Petviashvili equation},
         J. Math. Phys., \textbf{54} (2013), No.043511 (20pp).

\bibitem[Zhang20]{Zhang20}
D.J. Zhang,
\textit{Wronskian solutions of integrable systems}, In: \textit{Nonlinear systems and their remarkable mathematical structures},
 CRC Press, Boca Raton, FL, 2020, pp. 415--444.

\end{thebibliography}
\end{document}